\documentclass[11pt]{article}

\usepackage{geometry}
\usepackage{amssymb}
\usepackage{mathrsfs}
\usepackage{amsmath}
\usepackage[square]{natbib}
\usepackage{enumerate}

\usepackage[ruled]{algorithm2e}

\SetAlFnt{\small}
\SetAlCapFnt{\small}
\SetAlCapNameFnt{\small}
\SetAlCapHSkip{0pt}
\IncMargin{-\parindent}

\newtheorem{theorem}{Theorem}[section]

\newtheorem{example}{Example}[section]
\newtheorem{corollary}[theorem]{Corollary}
\newtheorem{lemma}[theorem]{Lemma}
\newtheorem{proposition}[theorem]{Proposition}
\newtheorem{claim}[theorem]{Claim}
\newtheorem{definition}[theorem]{Definition}
\newtheorem{remark}[theorem]{Remark}

\def\squarebox#1{\hbox to #1{\hfill\vbox to #1{\vfill}}}
\newcommand{\qed}{\hspace*{\fill}
\vbox{\hrule\hbox{\vrule\squarebox{.667em}\vrule}\hrule}\smallskip}

\newenvironment{proof}{\noindent{\bf Proof:~~}}{\(\qed\)}

\newcommand{\we}{competitive equilibrium}

\newcommand{\promising}{high-demand }

\newcommand{\fcwe}{competitive bundling equilibrium}
\newcommand{\fcwa}{competitive bundling equilibria}

\newcommand{\FCWE}{Competitive Bundling Equilibrium}
\newcommand{\pcbe}{partial \fcwe{}}
\newcommand{\Pcbe}{Partial \fcwe{}}
\newcommand{\opt}{\operatorname{OPT}}
\newcommand{\alg}{\operatorname{ALG}}
\newcommand{\poly}{\operatorname{poly}}

\newcommand{\bikh}{non-linear pricing equilibrium}
\newcommand{\bika}{non-linear pricing equilibria}

\newcommand{\BIKH}{Non-Linear Pricing Equilibrium}
\newcommand{\CAPTWO}{\operatorname{CAP2}}

\newcommand{\cI}{\mathcal{I}_M}
\newcommand{\cB}{\mathcal{B}}

\newcommand{\xa}{x_a}
\newcommand{\xb}{x_b}
\newcommand{\xab}{x_{ab}}
\newcommand{\ya}{y_a}
\newcommand{\yb}{y_b}
\newcommand{\yab}{y_{ab}}

\begin{document}

\title{Welfare and Revenue Guarantees for \\Competitive Bundling Equilibrium%
\thanks{We are grateful to Paul Milgrom, David Parkes and Tim Roughgarden for helpful discussions. Part of this work was done while authors 2-4 were visiting Microsoft Research Herzliya, Israel. I. Talgam-Cohen is supported by the Hsieh Family Stanford Interdisciplinary Graduate Fellowship. O. Weinstein is supported by the Simons Fellowship in Theoretical Computer Science.}
}

\author{Shahar Dobzinski\\
Weizmann Institute of Science
\and
Michal Feldman\\
Tel-Aviv University
\and
Inbal Talgam-Cohen\\
Stanford University
\and
Omri Weinstein\\
Princeton University}


\maketitle

\begin{abstract}
We study equilibria of markets with $m$ heterogeneous indivisible goods and $n$ consumers with combinatorial preferences. It is well known that a competitive equilibrium is not guaranteed to exist when valuations are not gross substitutes. Given the widespread use of bundling in real-life markets, we study its role as a stabilizing and coordinating device by considering the notion of \emph{competitive bundling equilibrium}: a competitive equilibrium over the market induced by partitioning the goods for sale into fixed bundles. Compared to other equilibrium concepts involving bundles, this notion has the advantage of simulatneous succinctness ($O(m)$ prices) and market clearance. 

Our first set of results concern welfare guarantees. We show that in markets where consumers care only about the number of goods they receive (known as multi-unit or homogeneous markets), even in the presence of complementarities, there always exists a competitive bundling equilibrium that guarantees a logarithmic fraction of the optimal welfare, and this guarantee is tight. We also establish non-trivial welfare guarantees for general markets, two-consumer markets, and markets where the consumer valuations are additive up to a fixed budget (budget-additive).

Our second set of results concern revenue guarantees. Motivated by the fact that the revenue extracted in a standard competitive equilibrium
may be zero (even with simple unit-demand consumers), we show that for natural subclasses of gross substitutes valuations, there always exists a competitive bundling equilibrium that extracts a logarithmic fraction of the optimal welfare, and this guarantee is tight. The notion of competitive bundling equilibrium can thus be useful even in markets which possess a standard competitive equilibrium.
\end{abstract}


\section{Introduction}
Competitive equilibria play a fundamental role in market theory and design -- they capture the market's steady states, in which each participant maximizes his profit at equilibrium prices, and supply equals demand. In this paper we focus on single-seller combinatorial markets, which consist of a set $M$ of $m$ heterogeneous goods, and a set $N$ of $n$ consumers. Each consumer $i$ has a valuation $v_i:2^M\rightarrow \mathbb R^+$ over bundles of goods. The standard assumptions are that each $v_i$ is normalized ($v_i(\emptyset)=0$) and monotone non-decreasing. In such markets, a competitive equilibrium is simply an allocation of the goods to the consumers, denoted by $(S_1,\ldots, S_n)$, and supporting item prices, denoted by $p_j$ for good $j$, such that:
\begin{enumerate}
\item \textbf{Profit Maximization:} The profit of every consumer $i$ is maximized by his allocation $S_i$; i.e., for every alternative set of goods $T$, $v_i(S_i)-\sum_{j\in S_i}p_j \geq v_i(T)-\sum_{j\in T}{p_j}$.
\item \textbf{Market Clearance:} All items are allocated; i.e., $\bigcup_iS_i=M$.%
\end{enumerate}
Unfortunately, a competitive equilibrium is guaranteed to exist only in limited classes of combinatorial markets, most notably the class in which all valuations are gross substitutes%
\footnote{See Section \ref{sec:two-buyers} for a definition of gross substitutes; see \citep{SY06,COP13} for other special cases in which a competitive equilibrium is guaranteed to exist.} %
\citep{GS99,Mil04}.

Implicit in the definitions of a market and an equilibrium is the assumption that the goods are exogenously determined and indivisible. Yet in many markets this is not the case, and what is being sold on the market are actually divisible bundles of indivisible items \citep[cf.,][]{LM-10}. For example, there is no inherent reason for beer to be sold in $6$-packs rather than, say, $8$-packs. Bundling is ubiquitous in real-life markets: It is a well-known method for revenue extraction \citep{MV06}; for instance, many airlines set the price of a round-trip to be equal to the price of a one-way air ticket. It is also a common mean for avoiding the exposure problem related to complementarities; for example, in the online market for concert ticket resale \texttt{StubHub.com}, a seller holding two tickets may prohibit their separate resale, so that he may still enjoy the show with a friend in case there is no demand for both tickets. 

\emph{In this paper we study the role of bundling in market stabilization.} Bundling introduces new equilibria into the market, and thus can recover stability in markets that lack a competitive equilibrium. It can also give rise to stable states with better revenue properties. The main challenge is whether ``good'' bundlings exist, i.e., those which result in (approximately) optimal social efficiency and/or revenue extraction.

\subsection*{Equilibrium Concept and Related Work}

There is a rich literature concerning extensions of competitive equilibrium that allow bundling.

\citet{BO02} study competitive equilibria over bundles supported by $2^m$ non-linear and non-anonymous prices per consumer, and use an LP-based approach to characterize them. Auctions that reach such equilibria are studied by \citet{PU00,AM-02}. Non-linear but anonymous supporting prices, as well as auctions that reach them, exist in special cases, notably superadditive valuations \citep{BO02,LP09,SY14}. See also \citep{Par99,parkes02,Par06,Voh11}.

In this paper we focus on perhaps the simplest possible extension, in which supporting prices are linear and anonymous. A \emph{competitive bundling equilibrium} consists of a partition of the goods into bundles, denoted by $\mathcal{B}=(B_1,\ldots, B_{m'})$ and referred to as a \emph{bundling}, in addition to an allocation $(S_1,\ldots, S_n)$ of the bundles to the consumers , and a price $p_{B_j}$ for each bundle $B_j$, such that:
\begin{enumerate}
\item \textbf{Profit Maximization:} For every consumer $i$ and alternative set of bundles $T$, $v_i(S_i)-\sum_{B_j\in S_i}p_{B_j} \geq v_i(T)-\sum_{B_j\in T}p_{B_j}$.
\item \textbf{Market Clearance:} $\bigcup_iS_i=M$.
\end{enumerate}
(For simplicity we use here the notation $v_i(S_i)$ to refer to $v_i(\bigcup_{B_j\in S_i}B_j)$, and $\bigcup_iS_i$ to refer to $\bigcup_i\bigcup_{B_j\in S_i}B_j$.) 

Observe that unlike competitive equilibrium, a competitive bundling equilibrium always exists by bundling all goods together, but the welfare it achieves may reach only a $1/n$-fraction of the optimal unconstrained welfare. Hence, in this paper we are mainly occupied with seeking better competitive bundling equilibria.

Relevant to our work is a recent paper by \citet{FGL}, which shows that if one ignores the market clearance requirement, there always exist a bundling and (anonymous and linear) prices that achieve a striking 2-approximation to the optimal welfare. Their solution concept, termed \emph{combinatorial Walrasian equilibrium}, obeys only the profit maximization requirement, and so is closer to an algorithmic pricing solution than to a classic, supply-equals-demand equilibrium. This prompts them to pose the following question, which we address in this paper by quantifying the effect of the market clearance requirement: ``An important question is whether our results extend to the stronger equilibrium notion [...] with market clearance.'' Despite the dissimilarities, we are able to adapt in this paper some of the technical constructions of Feldman et al.

The market clearance problem is also considered by \citet{FL}, but with respect only to two restricted classes of valuations -- a strict subclass of budget-additive valuations, and the class of superaddivite valuations. For the former we provide a general treatment in Section \ref{sec:budget-additive}. For the latter, we re-derive their result as a corollary of a more general LP-based argument that appears in Section \ref{sec:LP}: \citet{Par99} observes that an LP formulation introduced by \citet{BO02} has an integrality gap of $1$ for superadditive valuations; we show that this implies the existence of a fully-efficient competitive bundling equilibrium.

\subsection*{Disadvantages of Competitive Bundling Equilibrium}

We stress we have no claim that competitive bundling equilibrium is a universally applicable solution concept. In particular, it may be inappropriate in three kinds of markets: 
\begin{enumerate}
\item Markets in which consumers are able and willing to break their bundles and resell the goods separately -- in such markets, \fcwe{} is too weak a notion to establish stability;
\item Markets which can reach a stable state even without clearing -- in such markets, \fcwe{} is too strong a notion to capture all stable outcomes;
\item Markets in which different producers are not able to bundle their products together (reducing the notion of \fcwe{} to the standard one of competitive equilibrium).
\end{enumerate}
Like the seminal model of \citet{BO02}, our model assumes resale restrictions to enforce bundling (similar to their ``crates cannot be opened'' condition), as well as market clearance even in single-seller markets, and the ability to  bundle items arbitrarily. However, as we now explain, both models are relevant in a wide variety of markets despite these restrictions. 

First, bundling is often legally or effectively irrevocable. The market for air tickets mentioned earlier is one example, as are retail markets in which bundles are explicitely marked as ``not for individual sale''. Another example is markets in which the physical packaging plays a role, such as sterilized infant products. As yet another example consider amusement park passes -- Disney sells bundles of several day passes which are activated upon first enterance, when a fingerprint is taken from the visitor; the rest of the passes can then only be used by the same visitor. 

Second, in many markets the requirement of market clearance is crucial for stability. The standard argument is that if the market is not cleared, competing producers have an incentive to undercut prices, leaving the market unstable \citep[for a thorough discussion see][Section 10.B]{MWG95}. Market clearance also plays a role when there is a single seller: While in some cases a monopolist can create a stable outcome in which goods that were already produced are not sold, despite unsaturated demand, in many other cases this would not be possible. Consider for example a governmental seller of a scarce resource such as spectrum or land; an outcome where the government withholds unutilized but demanded resources is unreasonable and unstable. As another example, consider a private, for-profit seller who cannot credibly commit to withholding goods. For instance, a seller of concert tickets may attempt to create demand and extract more revenue by withholding half the tickets, but once the designated tickets are sold for a high price, he has incentive to simply sell more tickets at a lower price (say, at the door). Many consumers, in turn, would refrain from buying the high-priced tickets in the first place. Note that even the extreme method of physically destroying retail goods as a means of credible commitment is often impractical \citep[see consumer backlash over practice of destroying unsold clothing,][]{Dwy10}.

Third, bundles of goods belonging to different producers are not uncommon in practice. For example, a travel website can offer bundles of air tickets, hotel rooms and car rental. An interesting open question is what are the market processes leading to such bundling \citep[cf.,][]{HKN+14}. 

\subsection*{Our Results}

We analyze two different aspects of competitive bundling equilibria. First, the existence of approximately welfare-maximizing equilibria for several classes of well-studied valuations \citep[for a taxonomy see, e.g.,][]{BN07}. Next, we turn our attention to revenue guarantees: can the seller bundle the goods in a way that will guarantee approximately-optimal revenue, compared to the full-welfare benchmark? In our results, the approximation parameter%
\footnote{Interdisciplinary readers will recall that an approximation parameter of $\alpha$ means the \fcwe{} achieves at least $\opt/\alpha$, where $\opt$ is achieved by the objective-maximizing allocation, unrestricted to equilibrium allocations. A smaller $\alpha$ means better approximation.} %
will often depend (often logarithmically) on the length of the market's shorter side%
\footnote{Note that if one side of the market is of constant length we achieve a constant approximation.} %
$\mu$, where
\begin{align}
\mu=\min\{n,m\}.
\end{align}

\subsubsection*{Welfare Maximization}

The positive results described below establish, for different settings, the existence of \emph{a} \fcwe{} with good welfare guarantees; note that by the first welfare theorem applied to the post-bundling market, \emph{all} competitive bundling equilibria with this bundling have these guarantees.

We begin our overview with multi-unit settings, i.e., markets in which all goods are identical, and consumers' values only depend on the number of units they receive. A classic result of \citet{vic-61} shows that if the consumers' valuations exhibit decreasing marginal utilities, then there always exists a competitive equilibrium. We are able to give a complete analysis for competitive bundling equilibria in multi-unit settings, without the assumption of decreasing marginal utilities, thus accommodating the notorious case of markets with complements. We prove (Section \ref{sec:multi-unit}):

\vspace{0.1in}\noindent \textbf{Theorem: } For every multi-unit market there exists a competitive bundling equilibrium that provides an $O(\log \mu)$-approximation to the optimal welfare. On the flip side, if even a single valuation is \emph{subadditive} rather than having decreasing marginal utilities, it may be the case that the approximation ratio of every competitive bundling equilibrium is $\Omega(\log \mu)$.

\vspace{0.1in}\noindent The second half of the above theorem implies that even for slightly more complicated valuations than those considered by Vickrey, it is not only the case that a (standard) competitive equilibrium may not exist, but moreover possibly every competitive bundling equilibrium does not provide a constant factor of the optimal welfare. We also show (Section \ref{apx:randomness}) that even allowing for randomization by introducing lotteries into the market is not guaranteed to improve the approximation factor; this is by an interesting connection to optimal mechanism design.

We then consider combinatorial markets with heterogeneous goods, for which we obtain the following results. We warm up by analyzing the two-consumer case, showing that (Section \ref{sec:two-buyers}):

\vspace{0.1in}\noindent \textbf{Theorem: }For every combinatorial market with $n=2$ subadditive consumers, there exists a competitive bundling equilibrium that realizes an efficient allocation. For $n=2$ general valuations, there exists a competitive bundling equilibrium that provides a $\frac 3 2$-approximation to the optimal welfare, and this factor is tight.

\vspace{0.1in}\noindent We prove the above theorem by reducing the market to a new market with only $m=2$ goods. In such a market, the set of subadditive valuations coincides with the set of gross substitutes valuations, for which a competitive equilibrium is guaranteed to exist. By the first welfare theorem this equilibrium is also efficient. For $n>2$ consumers with general valuations, we obtain a weaker bound (Section \ref{sec-general}):

\vspace{0.1in}\noindent \textbf{Theorem: }For every combinatorial market with $n$ general valuations and $m$ goods, there exists a competitive bundling equilibrium that provides an $\tilde{O}(\sqrt \mu)$-approximation to the optimal welfare.

\vspace{0.1in}\noindent Our $\Omega(\log n)$ lower bound for multi-unit markets clearly holds for general valuations, but we do not know if competitive bundling equilibria can guarantee a better approximation ratio in this setting. We also do not know whether it is possible to obtain better approximation ratios if the valuations are known to be subadditive or even submodular -- these are left as open questions. However, for \emph{budget-additive} valuations -- an interesting subclass of submodular valuations -- we are able to show the following (Section \ref{sec:budget-additive}):

\vspace{0.1in}\noindent \textbf{Theorem: }For every combinatorial market with $n$ budget-additive valuations and $m$ goods, there exists a competitive bundling equilibrium that provides an $O(\log \mu)$-approximation to the optimal welfare.

\vspace{0.1in}\noindent We do not know whether this bound is tight (the $\Omega(\log \mu)$ lower bound does not hold here since it requires a non-submodular valuation). However, we show an instance in which no competitive bundling equilibrium can guarantee an approximation ratio better than $\frac 5 4$. A lower bound of $\frac 8 7$ was previously established by \citet{FL}.

While our main focus is on the basic question of existence of competitive bundling equilibria with good welfare guarantees, an additional interesting problem is to develop algorithms for finding such equilibria. All of our results are constructive, and for most we present efficient algorithms that find an equilibrium with welfare close to the existential guarantees. Yet, for several settings such as combinatorial markets with budget-additive consumers, it is an open question to design such efficient algorithms.

\subsubsection*{Revenue Maximization}

Almost all of our bounds on the welfare are actually bounds on the revenue. For example, the $\tilde{O}(\sqrt \mu)$-approximation for welfare in markets with general valuations actually guarantees a competitive bundling equilibrium which extracts at least $\tilde{\Omega}(1/\sqrt \mu)$ of the optimal revenue. Since better bounds on revenue for the classes of valuations we consider immediately imply better bounds on the welfare, we focus on the class of gross substitutes valuations. When valuations are gross substitutes, a competitive equilibrium always exists, and by the first welfare theorem the equilibrium allocation is efficient. However, it is easy to see there are instances in which the revenue in every competitive equilibrium is $0$ (even for unit-demand valuations). 

A natural approach to revenue extraction is limiting the supply, which has the effect of increasing competition among the consumers \citep{RTY12}. However, recall we are dealing with markets that are unstable unless they clear; bundling should thus be done carefully to ensure market clearance. This is a non-trivial task since the class of gross substitutes is not closed under bundling. We focus on markets in which consumers have \emph{weighted matroid rank} valuations, an important subclass of gross substitutes valuations. We show that (Section \ref{sec:revenue}):

\vspace{0.1in}\noindent \textbf{Theorem: }Consider a combinatorial market with $m$ goods. For the following valuation settings, there exists a competitive bundling equilibrium that extracts as revenue an $\Omega(1/\log \mu)$-fraction of the optimal welfare:
\begin{enumerate}
\item All $n$ valuations are weighted rank functions of uniform matroids.
\item All $n$ valuations are weighted rank functions of a common matroid, but with different weights.
\end{enumerate}

\noindent Notice that the optimal welfare is an obvious upper bound on the possible revenue that one can generate from the market. We can show that both bounds are tight: even if the valuations are unit-demand (a special case of both settings), no competitive bundling equilibrium can guarantee better revenue. An interesting question is to determine whether it is possible to extend our revenue results to the class of all gross substitutes valuations.

\subsection*{Organization}
As a warm-up and to introduce some of our ideas, in Section \ref{sec:two-buyers} we give a complete analysis for two-consumer markets. Section \ref{sec:prep} develops several useful technical tools for welfare analysis. Section \ref{sec:multi-unit} gives a complete analysis of welfare guarantees in multi-unit markets. Section \ref{sec-general} discusses general markets, and Section \ref{sec:budget-additive} analyzes the case of budget-additive valuations. We proceed to address revenue guarantees in Section \ref{sec:revenue}. In Section \ref{sec:LP} we discuss the relationship between competitive bundling equilibrium and the LP-based solution of \citet{BO02}. Finally, in Section \ref{apx:randomness} we discuss whether randomization can help circumvent our impossibility results. Some of the proofs are deferred to Appendices \ref{apx:two-buyers} to \ref{apx:revenue}.

\section{Warm-Up: Welfare Maximization with Two Consumers}
\label{sec:two-buyers}

In this warm-up section we provide an exact analysis of markets with two consumers. We show that if both valuations are subadditive then there exists a \fcwe{} that is fully efficient. On the other hand, if at least one valuation is not subadditive, we can only guarantee an 
approximation ratio of $3/2$. We also show that this last bound is tight. 

Recall that a valuation $v$ is subadditive if for every two sets of items $T,U$ it holds that $v(T)+v(U)\ge v(T\cup U)$. A valuation $v$ is gross substitutes if for
every item-price vectors $\vec{q}\ge\vec{p}$, for every set of items $T$ in the demand set%
\footnote{Recall that given a valuation $v$ and prices $\vec p$, an item set $T$ is in the demand set if it maximizes the payoff, i.e., $v(T)-\sum_{j\in T}p_j= \max_{U\subseteq M}\{v(U)-\sum_{j\in U}p_j\}$.} %
given prices $\vec{p}$, there exists a set $U$ in the demand set given prices $\vec{q}$, such that all items $j\in T$ whose price did not increase ($q_j=p_j$) belong to $U$. 

\begin{lemma}\label{lem:subadditive}
For every combinatorial market with $n$ subadditive consumers, there exists a \fcwe{} that provides an $\frac n 2$-approximation to the optimal welfare.
\end{lemma}
\begin{proof}
Fix an optimal allocation of all items $(O_1,\ldots , O_n)$. Without loss of generality assume that $v_1(O_1)\geq v_2(O_2) \geq \ldots \geq v_n(O_n)$. Consider the bundling $\mathcal{B}=(B_1,B_2)$ where $B_1=O_1$, $B_2=\bigcup_{i\geq 2}O_i$, which creates a new market instance with two items. It is well known that with two items, the class of subadditive valuations coincides with the class of gross substitutes valuations \cite{KC82}.
Therefore, the new market instance admits a competitive equilibrium, which is a special case of a \fcwe{} and which maximizes the social welfare by the first welfare theorem. The welfare of this \fcwe{} is at least $v_1(O_1)+v_2(O_2)$, 
which in turn is at least $\frac 2 n \cdot \Sigma_iv_i(O_i)$.
\end{proof}

\begin{corollary}
For every combinatorial market with $n=2$ subadditive consumers, there exists a \fcwe{} that obtains the optimal welfare.
\end{corollary}
We would like to show that this result cannot be extended to instances with two consumers and general valuations. To prove this, we require a version of the second welfare theorem that asserts that a competitive equilibrium exists in a market if and only if the following LP (the ``configuration LP'') has an optimal integral solution (see, e.g., \cite{BM97}):

\emph{Maximize:} $\Sigma_{i,S}x_{i,S}v_i(S)$

\emph{Subject to:}
\begin{itemize}
\item For each item $j$: $\Sigma_{i,S|j\in S}x_{i,S}\leq 1$
\item for each bidder $i$: $\Sigma_{S}x_{i,S}\leq 1$
\item for each $i$, $S$: $x_{i,S}\geq 0$
\end{itemize}

We are now ready to prove that:

\begin{proposition}\label{prop_two_consumers_lb}
There exists an instance with two consumers and two items such that every \fcwe{} obtains at most $\frac 2 3$ of the optimal welfare.
\end{proposition}
\begin{proof}
Consider an instance with $m=2$ items and $n=2$ consumers, where consumer $1$ has value $1$ for any single item, and value $2+\epsilon$ if he gets both items, where $\epsilon>0$; consumer $2$ is unit-demand, with value $2$ for any item he gets. The optimal welfare equals~$3$. 

We will show that there does not exist a \we{} in this market. Thus, in any \fcwe{} the two items must be bundled together. The best solution in this case is to give the bundle to consumer $1$, achieving welfare $2+\epsilon$. The approximation ratio of this \fcwe{} approaches $3/2$ as $\epsilon$ approaches $0$.

Although it is not hard to directly prove that there is no \we{} in this market, by the version of the second welfare theorem stated above  it is sufficient to show a fractional solution to the \emph{configuration LP} that exceeds 3, which is the value of the optimal integral solution. Consider the fractional solution which assigns half of the grand bundle to consumer $1$, and assigns half of each individual item to consumer $2$. This value of this fractional solution is $3+\epsilon/2$, which implies that a competitive equilibrium does not exist. This finishes the proof
\end{proof}
The impossibility result implied by Proposition \ref{prop_two_consumers_lb} trivially extends to more consumers and items (by adding consumers with zero valuations and items with zero marginal contribution), and is tight as the next proposition asserts.

\begin{proposition}
\label{prop_two_consumers_ub}
For every combinatorial market with $n=2$ consumers and general valuations, there exists a \fcwe{} that provides a $\frac 3 2$-approximation to the optimal welfare.
\end{proposition} 
The proof of Proposition \ref{prop_two_consumers_ub} involves a careful case-by-case analysis and is deferred it to Appendix \ref{apx:two-buyers}.

\section{General Welfare Maximization: Preparations}
\label{sec:prep}

In the next sections we study welfare-maximizing \fcwa{} in different combinatorial markets. The goal of this section is to prepare our main ``working horses'' for this task, namely Lemma \ref{lem:fcwe-from-partial} and Lemma \ref{lem:promising}. The proofs of both lemmas use the following result from \cite{FGL}:

\begin{theorem}[essentially \cite{FGL}]\label{lemma-FGL}
In a combinatorial market with general, possibly non-monotone valuations, let $\mathcal B$ be a bundling with bundle prices $\vec{p}$. Then there exist a bundling $\mathcal B'$ over bundles in $\mathcal B$, prices $\vec{p'}$ over $\mathcal B'$ and an allocation $(S_1,\ldots, S_n)$ over $\mathcal B'$ such that:
\begin{enumerate}
\item\label{item-FGL-prices-up} For every $i$, let $T_i\subseteq \mathcal B$ be the set of bundles that $S_i$ is combined from, then $p'_i\geq \Sigma_{j\in T_i}p_j$.
\item\label{item-FGL-unallocated} If some $B\in \mathcal B'$ is not allocated then $B\in \mathcal B$ and $p'_B=p_B$.
\item\label{item-FGL-profit-maximization} For each consumer $i$ we have that $S_i$ is in $i$'s demand set given prices $p'$, i.e., $v_i(S_i)-\Sigma_{B\in S_i}p'_B \geq v_i(T)-\Sigma_{B\in T}p'_B$ for all $T\subseteq\mathcal B'$.
\end{enumerate}
Furthermore, the bundling, prices and allocation can be found using $\poly(m,n)$ demand queries given $\mathcal B$ and $\vec{p}$.%
\footnote{For given prices, a demand query for a certain valuation returns a member of its demand set under these prices.}
\end{theorem}
Note that the bundling, prices and allocation guaranteed by Theorem \ref{lemma-FGL} do not in general form a \fcwe{}, since market clearance is not guaranteed, i.e., it does not necessarily hold that $\cup_i S_i=M$.

Our first lemma shows how to ``lift'' a \emph{partial} \fcwe{} (which is easier to construct) to a proper \fcwe{}. 

\begin{definition}
An allocation $S$ over a bundling $\mathcal B$ and a vector of bundle prices $\vec{p}$ form a \emph{\pcbe{}} with respect to some set of consumers $N'$, if they constitute a \fcwe{} with respect to $N'$.
\end{definition}
In other words, the allocation, bundling and prices of a partial \fcwe{} constitute a \fcwe{} assuming that the valuation of every consumer not in $N'$ is identically $0$.

\begin{lemma}
\label{lem:fcwe-from-partial}
Consider an allocation $(S_1,\ldots, S_n)$ over a bundling $\mathcal B$ and bundle prices $\vec p$, which form a \pcbe{} with respect to $N'\subseteq N$. Then there exists a bundling $\mathcal B'$ over bundles in $\mathcal B$, prices $\vec{p'}$ over $\mathcal B'$ and an allocation $(S'_1,\ldots, S'_n)$ over $\mathcal B'$, which form a \fcwe{} with welfare at least as high as the revenue of the \pcbe{}:
$$
\sum_{i\in N} v_i(S'_i)\geq \sum_{B\in\mathcal{B}} p_B.
$$
Furthermore, this \fcwe{} can be found using $\poly(m,n)$ demand queries given the \pcbe{}.
\end{lemma}
The proof of Lemma \ref{lem:fcwe-from-partial} appears in Appendix \ref{apx:prep}. It makes use of the fact that Theorem \ref{lemma-FGL} does not even require the valuations to be monotone.

\vspace{0.1in}We also use a different method to prove the existence of a \fcwe{} with good properties, as follows. 

\begin{definition}
A \emph{\promising priced bundling} consists of a bundling $\mathcal{B}$ and bundle prices $\vec{p}$, such that for every bundle $B\in\mathcal{B}$, there is a set $N_B\subseteq N$ of at least $|\mathcal{B}|$ consumers for which $B$ is profitable: 
$$
\forall i\in N_B\;:\;v_i(B)-p_B > 0.
$$ 
\end{definition}

\begin{lemma}
\label{lem:promising}
Consider a \promising priced bundling $\mathcal{B}$ with bundle prices $\vec p$. Then there exists a bundling $\mathcal B'$ over bundles in $\mathcal B$, prices $\vec{p'}$ over $\mathcal B'$ and an allocation $(S'_1,\ldots, S'_n)$ over $\mathcal B'$, which form a \fcwe{} with welfare at least as high as the aggregate price of the \promising priced bundling:
$$
\sum_{i\in N} v_i(S'_i)\geq \sum_{B\in\mathcal{B}}p_B.
$$
Furthermore, this \fcwe{} can be found using $\poly(m,n)$ demand queries given $\mathcal{B}$ and $\vec{p}$.
\end{lemma}
The proof of Lemma \ref{lem:promising} appears in Appendix \ref{apx:prep}.

\vspace{0.1in}Finally, the following simple lemma will be useful for our analysis. 

\begin{lemma}
\label{lem:logarithmic}
For every allocation $S=(S_1,\dots, S_n)$ of items $M$, there exists a value $v$ and an allocation $S'=(S'_1,\dots, S'_n)$ of $M'\subseteq M$ such that for every $S'_i\ne\emptyset$, $v_i(S'_i)\in [v,2v)$, and the welfare of $S'$ is a logarithmic approximation to the welfare of $S$:
$$
\sum_{i\in N}{v_i(S'_i)} \ge \frac {1} {\alpha} \sum_{i\in N}{v_i(S_i)},
$$
where $\alpha=2(\log\mu+2)$. Furthermore, $v$ and $S'$ can be found given $S$ in $\poly(m,n)$ time using value queries.
\end{lemma}
The proof of Lemma \ref{lem:logarithmic} appears in Appendix \ref{apx:prep}.

\section{Welfare Maximization in Multi-Unit Markets}
\label{sec:multi-unit}

In this section we focus on multi-unit markets, and show existence of a \fcwe{} with a logarithmic approximation ratio to the optimal social welfare, denoted by $\opt$. Our positive result in Section \ref{sub:multi-pos} holds for any multi-item valuations, i.e., units may be a mixture of substitutes and complements with respect to each other. We then show in Section \ref{sub:multi-tight} that this result is tight, in the sense that there are multi-unit settings for which no \fcwe{} exists with welfare better than a logarithmic fraction of $\opt$. In fact, this holds even if the valuations are all subadditive.

\subsection{A Logarithmic Approximation}
\label{sub:multi-pos}

We now state and prove our main positive result for this section.

\begin{theorem}
\label{thm:multi-positive}
For every multi-unit market with $n$ consumers and $m$ items there exists a \fcwe{} that provides an $O(\log \mu)$-approximation to the optimal welfare., where $\mu=\min\{m,n\}$.
\end{theorem}

\begin{proof}
Consider a welfare-optimal allocation $(O_1,\ldots, O_n)$ of the items. We show there exists a \promising priced bundling whose aggregate price is an $O(\log \mu)$-approximation to $\opt$; the proof is then complete by applying Lemma \ref{lem:promising}.

We begin by applying Lemma \ref{lem:logarithmic} to show that there exist a value $v$ and an allocation $(O'_1,\ldots, O'_n)$ of an item subset $M'\subseteq M$, such that (i) for every consumer $i$ with non-empty allocation, $v_i(O'_i)\in[v,2v)$; (ii) a logarithmic fraction of the welfare is preserved: 
$$
\sum{v_i(O'_i)} \ge \opt/\Theta(\log\mu).
$$ 
Without loss of generality assume $|O'_1| \geq \dots \geq |O'_n|$, and let $n'$ be the largest index such that $|O'_{n'}|>0$. If $n'=1$, the grand bundle can be allocated to consumer $1$ for price $v_1(M)$ and we are done, so assume from now on $n'>1$. We now use allocation $O'$ to show the existence of the \promising priced bundling. 

Let $\mathcal{B}$ be a partition of all items into $k=\lfloor n'/2\rfloor$ bundles of equal size (if $k$ does not divide $m$, place leftover items in one of the bundles arbitrarily). Note that each bundle $B\in\mathcal{B}$ has size at least $|O'_{k+1}|$, by the following argument:
$$
m \ge \sum_{i\le k}{|O'_i|}\ge k|O'_k| \implies \lfloor m/k \rfloor \ge |O'_k| \ge |O'_{k+1}|.
$$ 
Set the price of every bundle to be $p_B=v-\epsilon$. All $k$ bundles are profitable for consumers $k+1,\dots,n'$, and since there are at least $k$ such consumers we have a \promising priced bundling. It remains to show that the aggregate price $k(v-\epsilon)$ is a logarithmic fraction of $\opt$:
$$
k(v-\epsilon) \ge 2v(2k+1)/10 \ge 2vn'/10 > \frac 1 {10} \sum{v_i(O'_i)} \geq \opt / \Theta(\log\mu),
$$ 
where the first inequality is by choosing $\epsilon$ sufficiently small and the third inequality uses that $v_i(O'_i)<2v$.
\end{proof}

\begin{remark}
The above proof provides an algorithm for finding the guaranteed \fcwe{}, which runs in $\poly(m,n)$ time and makes that many demand queries. We now show how to transform it into an algorithm with $\poly(\log(m), n)$ running time, which uses only value queries and loses only an extra factor of $2$ in the approximation ratio (recall that even in a multi-unit setting, demand queries are known to be strictly stronger than value queries). This running time can be considered polynomial for the multi-unit setting.

The idea is to preprocess, by bundling the items into equal-sized bundles of size $m/n^2$ (ignoring leftovers). \citet{DN10} show that the optimal allocation of such bundles achieves a $2$-approximation to the optimal unconstrained welfare. All that is left to show is that demand queries for this new setting can be simulated by $\poly(\log(m),n)$ value queries. Since the total number of units in any bundle is now a multiple of $m/n^2$, one can use dynamic programming to simulate a demand query, using only $\poly(\log(m),n)$ value queries.
\end{remark}

\subsection{A Tight Impossibility Result}
\label{sub:multi-tight}

The next theorem shows that it is impossible to obtain an approximation ratio better than logarithmic by a \fcwe{}, even if all valuations are subadditive.

\begin{theorem}
\label{thm:multi-tight}
There exists a multi-unit market where $n=m$ and valuations are subadditive, such that every \fcwe{} has welfare that is a $1/\Omega(\log m)$-fraction of the optimal social welfare.
\end{theorem}
\begin{proof}
We construct the market as follows. Set $n=m$. For every $i\in\{2,\dots,m\}$, consumer $i$ has a unit-demand valuation with value $1/i$ (i.e., $\forall S\ne\emptyset : v_i(S)= 1/i$). Consumer $1$'s has the following valuation: for some small $\epsilon>0$, $v_1(S)=1+\epsilon$ for every $S$ of size $|S|<m$, and $v_1(M)=2+2\epsilon$. Notice that all valuations are subadditive.

The optimal allocation in this market allocates each of the consumers a single unit, achieving welfare of $\sum_i{1/i}\approx\ln m$. We will show that in every \fcwe{}, consumer $1$ is allocated all units, and thus the welfare is only $2+2\epsilon$.

For the rest of the proof we fix some \fcwe{} with allocation $S$ over bundling $\mathcal B$ with prices $\vec{p}$. Let $i'$ be the smallest index of a consumer who is allocated a bundle from $\mathcal{B}$.

\begin{claim}
\label{cla:equal-prices}
Every non-empty bundle $B\in\mathcal{B}$ has the same price $p_B=p$, and $p\leq 1/ i'$.
\end{claim}

\begin{proof}
Denote the bundle from $\cB$ allocated to consumer $i'$ by $B_{i'}$. We first observe that consumer $i'$ cannot be charged more than $p\leq 1/ i'$ for $B_{i'}$, otherwise his profit would be negative. Next, we observe that the price of any other non-empty bundle $B\in\mathcal{B}$ that is allocated to some consumer $i\neq i'$ must equal $p$: Since both $i$ and $i'$ are unit demand we have that $v_i(B)=v_i(B_{i'})$ and $v_{i'}(B)=v_{i'}(B_{i'})$. Thus if the price of one of these bundles is lower than the other, both $i$ and $i'$ will prefer to be allocated that bundle, and we are not at a \fcwe{}.
\end{proof}

Now, suppose towards contradiction that $i'>1$. By definition, at most $i'$ consumers are allocated in the \fcwe{}, and so by market clearance $|\mathcal{B}|\le i'$. By Claim \ref{cla:equal-prices}, the price of every $B\in\mathcal{B}$ is at most $1/ i'$. The total price for all bundles in $\mathcal{B}$ is therefore at most $1$.

Now consider consumer $1$, whose profit is at most his value $v_1(S_1)$. By the assumption that consumer $1$ is not allocated all units, his value is at most $ 1+\epsilon$. Observe that $v_1(\bigcup_{B\in\mathcal{B}}{B}) - \sum_{B\in\mathcal{B}}{p_B} \ge 2+2\epsilon-1=1+2\epsilon$, in other words, consumer $1$ strictly increases his profit by buying all bundles in $\mathcal{B}$. This is in contradiction to the profit maximization property of the \fcwe{}. This completes the proof of Theorem \ref{thm:multi-tight}.
\end{proof}

\section{Welfare Maximization in General Markets}\label{sec-general}

In this section we consider general markets with no assumptions on the valuations (except the standard ones of monotonicity and normalization). In Section \ref{sub:general-exist} we show that for every combinatorial market, there always exists a \fcwe{} whose allocation achieves a $\tilde O(\sqrt \mu)$-approximation to the optimal welfare. In Section \ref{sub:general-efficient} we address computational aspects.

\subsection{Existential Result}
\label{sub:general-exist}

The key lemma of this section is the following.

\begin{lemma}
\label{thm:general_efficient_convert}
Let $v$ be a value and let $S=(S_1,S_2,\ldots,S_r)$ be an allocation of a subset of the items to the first $r$ consumers, such that $\forall i\in[r]$ it holds that $S_i\ne\emptyset$ and $v_i(S_i)\in [v,2v)$. Then there is a \fcwe{} that achieves an $O(\sqrt {r})$-approximation to the welfare of $S$. 
\end{lemma}
Before proving this lemma we state and prove its corollary. Recall $\mu=\min\{m,n\}$.
\begin{corollary}
\label{thm:general-sqrt}
For every combinatorial market there exists a \fcwe{} with welfare that is a $\tilde{O}(\sqrt{\mu})$-approximation to the optimal welfare.  
\end{corollary}

\begin{proof}
Apply Lemma \ref{lem:logarithmic} to the welfare-optimal allocation $(O_1,\dots,O_n)$ to get a value $v$ and an allocation $S=(S_1,\dots,S_n)$ of items $M'\subseteq M$. Without loss of generality, assume that exactly the first $r$ allocated parts in $S$ are non-empty, and notice that $r$ must be $\le\mu$. We have that (i) for every consumer $i\in[r]$, $v_i(S_i)\in[v,2v)$; (ii) a logarithmic fraction of the welfare is preserved: 
\begin{align}
\sum_{i\in [r]}{v_i(S_i)} \ge \frac{1}{O(\log \mu)}\sum_{i}v_i(O_i).\label{eq:log-frac}
\end{align}
By applying Lemma \ref{thm:general_efficient_convert} to the value $v$ and allocation $S$, we get an $O(\sqrt{r})=O(\sqrt{\mu})$-approximation to the welfare of $S$. Combining this with Equation (\ref{eq:log-frac}) completes the proof. 
\end{proof}

We now prove the key lemma.

\begin{proof}[of Lemma \ref{thm:general_efficient_convert}]
We show how to construct a \promising priced bundling whose aggregate sum of prices is an $O(\sqrt{r})$-approximation to the welfare of $S$. The proof is then established by invoking Lemma \ref{lem:promising}. 

Begin by using the allocated parts $S_1,\dots,S_r$ to create new bundles, each but one containing $\lceil\sqrt{r}\rceil$ parts $S_i$, and the last one containing between $\lceil\sqrt{r}\rceil$ and $2\lceil\sqrt{r}\rceil$) of the parts. Add any items in $M$ that are unallocated in $S$ to an arbitrary bundle. Let $\mathcal{B}$ denote the resulting partition, and set a price $p_B=v$ for every $B\in \mathcal{B}$.

We now show that the pair $(\mathcal{B},\vec{p})$ is a \promising priced bundling. By construction, $|\mathcal{B}| \leq \sqrt{r}$, as the total number of bundles in $S$ was $r$, and every $B\in \mathcal{B}$ contains at least $\sqrt{r}$ such bundles. Furthermore, every bundle $B$ is profitable for at least $\sqrt{r}$ consumers (those who were originally allocated the parts of $S$ it contains, and value them at $\ge v$). Now, recall that
\begin{align}
\label{eq_s_v_ub}
\sum_{i\in [r]}{v_i(S_i)} \leq 2vr.
\end{align}
Since no bundle $B\in\mathcal{B}$ contains more than $2\lceil\sqrt{r}\rceil\le 2\sqrt{r}+2$ parts $S_i$, then $|\mathcal{B}|\ge r/(2\sqrt{r}+2)$, and so we have that
\begin{align*}
\sum_{B\in \mathcal{B}} p_B = |\mathcal{B}|v \geq \frac{rv}{2\sqrt{r} + 2} \geq 
\frac{1}{4\sqrt{r}+4} \sum_{i\in [r]}{v_i(S_i)},
\end{align*}
where the last inequality is by \eqref{eq_s_v_ub}. This completes the proof.
\end{proof}

\subsection{Efficiently Finding an Equilibrium}
\label{sub:general-efficient}
Corollary \ref{thm:general-sqrt} guarantees that in every market there exists a \fcwe{} that provides a $\tilde O(\sqrt m)$ approximation to the optimal welfare. It is not hard to see that, \emph{if one is given a welfare-optimal allocation}, it also indicates how to construct such a \fcwe{} in $\poly(m,n)$ time using demand queries. This is essentially due to the constructive nature and computational guarantees of Lemmas \ref{lem:promising} and \ref{lem:logarithmic}, which are the main tools applied in the existential proof. Yet finding a welfare-optimal allocation -- or even one that provides an approximation ratio better than $O(\sqrt{m})$ -- is impossible in general without exponential communication \citep[and references within]{BN07}. The purpose of this section is to show that a \fcwe{} with somewhat weaker welfare guarantees can be constructed ``from scratch'', in a computationally efficient way using $\poly(m,n)$ demand queries.%
\footnote{The results of the previous section also imply that given an $\alpha$-approximation algorithm for the allocation problem over a special class of valuations, it is computationally efficient (using demand queries) to find a \fcwe{} that provides a $\tilde O(\alpha \sqrt \mu)$-approximation. In particular, using the $2$-approximation algorithm of \citet{Fei09}, it is tractable to find a \fcwe{} that provides a $\tilde O(\sqrt \mu)$-approximation when all valuations are known to be subadditive.}

The main lemma of this section is as follows.

\begin{lemma}
\label{cl_sqrt_kn_greedy_approx} 
For every combinatorial market and bundle size $k\in[m]$, there is a computationally efficient algorithm that finds a \fcwe{} using $\poly(m,n)$ demand queries, such that for every allocation $S$ to consumers $N_k\subseteq N$ where $|S_i|\in[k,2k)$ for every $i\in N_k$, the welfare of the \fcwe{} is a $\tilde{O}(\sqrt{km})$-approximation to the welfare of $S$. 
\end{lemma}
Before proving this lemma we state and prove its corollary.
\begin{corollary}
\label{thm:general_efficient}
There is a computationally efficient algorithm that for every combinatorial market uses $\poly(m,n)$ demand queries to find a \fcwe{} with welfare that is a $\tilde{O}(m^{2/3})$-approximation to the optimal welfare.
\end{corollary}

\begin{proof}
We present an algorithm for constructing a high-welfare \fcwe{}. The algorithm performs an exhaustive search -- for every $k \in \{1,2,\ldots , \lceil m^{1/3}\rceil\}$, it applies Lemma \ref{cl_sqrt_kn_greedy_approx} to find a \fcwe{} for this $k$ (the choice of enumeration range becomes clear in the analysis below). It then compares the maximum welfare found by the exhaustive search to the welfare of the \fcwe{} in which the grand bundle is allocated to consumer $i^*$ with the highest value for it, charging $i^*$'s value as payment. The algorithm returns the \fcwe{} with the overall maximum welfare.

We now analyze the above algorithm and show it achieves a $\tilde{O}(m^{2/3})$-approximation; it is not hard to see the algorithm is computationally efficient and uses $\poly(m,n)$ demand queries. Consider a welfare-optimal allocation $O=(O_1,\ldots, O_n)$. A variation%
\footnote{The fact we consider a bundle size rather than value is the difference from Lemma \ref{lem:logarithmic}.} %
of Lemma \ref{lem:logarithmic} shows that there exists a bundle size $k^*\in[m]$, such that the set of consumers who are allocated roughly $k^*$ items $N_{k^*} := \{i \; : \; |O_i|\in [k^*,2k^*) \}$ accounts for a logarithmic fraction of the optimal welfare:
\begin{align*}
\sum_{i\in N_{k^*}} v_i(O_i) \ge \opt/O(\log m).
\end{align*}
Denote by $O'$ the allocation to consumer set $N_{k^*}$ (observe that $O, O'$ and $k^*$ are not assumed to be known by the algorithm, and are only used in its analysis). There are two possible cases: 

\begin{enumerate}
\item If $k^*\leq m^{1/3}$, running the algorithm guaranteed by Lemma \ref{cl_sqrt_kn_greedy_approx} will obtain a \fcwe{}, whose welfare is at least a $\tilde{O}(\sqrt{k^*\cdot m}) \leq \tilde{O}((m^{4/3})^{1/2}) = \tilde{O}(m^{2/3})$- approximation to the welfare of $O'$, and thus also to the welfare of $\opt$.
\item If $k^* > m^{1/3}$, consider the \fcwe{} resulting from bundling all items into a single grand bundle and allocating it to the consumer $i^*$ with the highest value for it. Since at most $m/k^*$ consumers can be allocated a bundle of size $\ge k^*$, the value of $i^*$ must be a $m/k^* < m/m^{1/3} = O(m^{2/3})$
-approximation to $\opt$.
\end{enumerate}
This completes the proof of the corollary.
\end{proof}

We now prove this section's key lemma.

\begin{proof}[of Lemma \ref{cl_sqrt_kn_greedy_approx}]
Given $k$, we define a greedy procedure $\mathcal{A}_k$ as follows. The procedure maintains a set $N'$ of consumers not yet allocated, and a set $M'$ of items not yet allocated. Initially, $N' \longleftarrow N$ and $M'\longleftarrow M$. Until $|M'| < k$, $\mathcal{A}_k$ repeats the following:
\begin{itemize}
\item Let 
$$
(i^*,A^*) = {\arg\max}_{i\in N',A\subseteq M' : |A|\in[k,2k)} \{v_i(A)\}
$$ 
be a pair of unallocated consumer and bundle of size $\in[k,2k)$ such that $v_{i^*}(A^*)$ is maximum. This is known to be efficiently computable with polynomially many demand queries \citep{BN05}.
\item Set $N' \longleftarrow N'\setminus \{i^*\}$, $M'\longleftarrow M'\setminus S^*$.
\end{itemize}

Without loss of generality let $A=(A_1, \ldots, A_r)$ be the final allocation reached by  $\mathcal{A}_k$ (where not all items are necessarily allocated). Observe that since $\mathcal{A}_k$ allocates a bundle of size at least $k$ in each iteration, it must be the case that $r \leq  m/k$. The proof of the following inequality is included below for completeness:
\begin{align}
\label{eq_k_approx_greedy}
\sum_{i\in[r]}{v_i(A_i)} \ge \frac{W}{2k},
\end{align}
where $W$ is the maximum welfare of an allocation $S$ such that $|S_i|\in[k,2k)$ for every non-empty bundle $S_i$.

We now apply Lemma \ref{lem:logarithmic} to $A$, resulting in a value $v$ and allocation $S'$ involving at most $r$ consumers, to which in turn we can apply Lemma \ref{thm:general_efficient_convert}. It is not hard to see that the application of these lemmas requires polynomial time and a polynomial number of demand queries.%
\footnote{Value queries, when used, can be simulated by demand queries \citep{BN07}.} %
The final output of Lemma \ref{thm:general_efficient_convert} is a \fcwe{} with welfare that is a $\tilde{O}(\sqrt{m/k})$-approximation to the welfare of $A$, and a $\tilde{O}(2k\sqrt{m/k})=\tilde{O}(\sqrt{km})$-approximation to $W$.

It is left to prove Inequality (\ref{eq_k_approx_greedy}). Consider the $t$th step of procedure $\mathcal{A}_k$. Let $i_1,\ldots, i_t$ be the consumers chosen in steps 1 to $t$ and let $M_t$ be the items allocated. Let $\alg_t$ be the aggregate welfare, that is, $\alg_t=\sum_{\ell\leq t}v_{i_\ell}(A_{i_\ell})$; set $\alg_0=0$. Let $W_t$ be the welfare of the allocation $S$ restricted to yet unallocated items $M\setminus M_t$ (note that $W_t$ is decreasing in $t$); set $W_0=W$. Observe that $\alg_{t}-\alg_{t-1}=v_{i_t}(A_{i_t})$. We claim that 
\begin{align}
\label{eq:step_t}
W_{t-1}-W_{t}\leq 2k(\alg_{t}-\alg_{t-1}).
\end{align}
Assuming this claim holds, starting with $t=1$ and summing over all steps of the procedure we get that $W\leq 2k\sum_{i\in[r]}{v_i(A_i)}$, as required. 

We now prove Inequality (\ref{eq:step_t}) for $t=1$; the proof for the following steps is similar and thus omitted. Consider consumer $i_1$ receiving $A_{i_1}$ in the first step of procedure $\mathcal{A}_k$. If $i_1$ was allocated $A_{i_1}$ in $S$ then we are done. Otherwise, we gain $v_{i_1}(A_{i_1})$, and lose at most the aggregate value of the $\le |A_{i_1}|\le 2k$ other consumers whose bundles in $S$ intersect with $A_{i_1}$. Since  $v_{i_1}(A_{i_1})$ is at least as large as the value of each of these consumers (else, procedure $\mathcal{A}_k$ would have chosen one of them as the first consumer to allocate to), Inequality (\ref{eq:step_t}) follows. This also completes the proof of Inequality (\ref{eq_k_approx_greedy}) and of Lemma \ref{cl_sqrt_kn_greedy_approx}.
\end{proof}

\section{Welfare Maximization with Budget-Additive Valuations}
\label{sec:budget-additive}

In Section \ref{sec-general} we showed that every combinatorial market admits a \fcwe{} that provides an approximation ratio of $\tilde O(\sqrt m)$ to the social welfare. A natural next step is to understand whether we can get better approximation ratios for specific subclasses. We make progress towards this goal by showing that if the valuations are all budget additive%
\footnote{A valuation $v$ is \emph{budget additive} if for every bundle $S$ we have that $v(S) = \min \{\sum_{j\in S} v(j), b\}$.} %
then we can get a logarithmic approximation. The best lower bound we currently know shows that no market with budget additive valuations can achieve an approximation ratio better than $\frac 5 4$.

\subsection{A Logarithmic Approximation}

\begin{theorem}
\label{thm:ba-positive}
In every combinatorial market with budget-additive valuations, there is a \fcwe{} with welfare that is an $O(\log m)$-approximation to the optimal welfare.
\end{theorem}

\begin{proof}
We construct an allocation as follows. For every valuation $v_i$ with budget $b_i$, define another valuation $v'_i$ where $v'_i(S)=\min \{\sum_{j\in S}v_i(j), 2b_i\}$. Recall the greedy algorithm of \cite{LLN06} for submodular valuations: consider the items one by one in an arbitrary order, and allocate each item to a consumer that maximizes the marginal value for it given the items he received until now. Let $A=(A_1,\ldots, A_n)$ be the allocation that the greedy algorithm produces when running on valuations $v'_1,\ldots, v'_n$. Let $\bar{E}:= \{i \mid v_i(A_i) = b_i \}$ be the set of consumers who have exhausted their budgets in the allocation $A$, and let $\bar{E}:= \{i \mid v_i(A_i) < b_i \}$ be the set of consumers who haven't exhausted their budgets.

\begin{claim}\label{cl_greedy_approx}
The allocation $A$ satisfies the following properties:
\begin{enumerate}
\item The welfare of $A$ is a $4$-approximation to the optimal welfare: $\sum_{i\in N} v_i(A_i) \geq \opt/4$.
\item For every consumer $i\in \bar{E}$ and item $j\in A_i$, it holds that $v_i(j)\ge\max_{i'\in \bar{E}}\{v_{i'}(j)\}$.
\end{enumerate}
\end{claim}
\begin{proof}[of Claim \ref{cl_greedy_approx}] 
For the first part, recall that the greedy algorithm provides a $2$-approximation to the optimal welfare with respect to the valuations $v'_i$ \citep{LLN06}. Since for each bidder $i$ the difference between $v_i$ and $v'_i$ is only the budget, we have that for every bundle $S$, $v_i(S)\leq v'_i(S)\leq 2v_i(S)$. Two direct consequences of this inequality are:
\begin{itemize}
\item The welfare of the optimal allocation with respect to the valuations $\{v_i\}$ is at most the welfare of the optimal allocation with respect to the valuations $\{v'_i\}$.
\item For every $i$, $v_i(A_i)\geq v'_i(A_i)/2$.
\end{itemize}
Together these imply the first part of the claim. The second part is a direct consequence of the definition of the greedy algorithm.
\end{proof}

We now distinguish between two cases:
\begin{enumerate}
\item \emph{Most of the welfare in allocation $A$ comes from consumers in $E$:}
$$
\sum_{i \in E} v_i(A_i) \geq \frac{1}{2}\sum_{i\in N} v_{i}(A_{i}).
$$
Assume that for all $i\in E$, $b_i \geq \sum_{i'\in E}v_{i'}(A_{i'}) / 2m$ -- this is without loss of generality as such consumers contribute at least a quarter of the total welfare of $A$ (c.f., proof of Lemma \ref{lem:logarithmic}). Divide the consumers in $E$ into bins according to their budgets, such that consumer $i$ is in bin $b$ if $b_i\in[b,2b)$. Notice there are at most $\log m+2$ bins. Denote the set of players in bin $b$ by $S_b$. By the pigeonhole principle, there must be one bin $b$ such that
\begin{align*}
\sum_{i\in S_b}b_{i} \geq \frac {\sum_{i \in E}b_i} {\log m+2} = \frac {\sum_{i\in E} v_{i}(A_{i})} {\log m +2}
\geq \frac {\sum_{i\in N} v_{i}(A_{i})} {4(\log m +2)},
\end{align*}
where the equality is by definition of $E$. On the other hand, by definition of $S_b$ we have $\sum_{i\in S_b}b_{i} \leq |S_b|\cdot 2b$, since $b\leq b_i < 2b$. Combining both inequalities
we get
\begin{align}\label{eq_rev_Sb}
&|S_b|\cdot b \geq \frac {\sum_{i\in N} v_i(A_i)} {8(\log m +2)} \geq \frac{\opt}{32(\log m +2)} = \frac{\opt}{O(\log m)},
\end{align}
where the second transition is by Claim \ref{cl_greedy_approx}.

To complete the construction, for every $i\in S_b$ place all items in $A_i$ in the same bundle, and add any remaining items (items that were allocated to other consumers or were not allocated at all) to an arbitrary bundle. Denote the resulting bundles by $B_i$. Now allocate $B_i$ to every consumer $i\in S_b$ for a uniform price of $p_{B_i}=b$. We claim we have constructed a \pcbe{} with respect to the consumers in $S_b$. To see this, note that $S_b\subseteq E$ and therefore $v_i(B_i) = b_i$, so $B_i$ maximizes the value of consumer $i\in S_b$. Since the price of all bundles is $b$, $B_i$ is contained in the demand set of consumer $i$. 

We now apply Lemma \ref{lem:fcwe-from-partial} to obtain a \fcwe{} with welfare at least
\begin{align*}\label{eq_rev}
\sum_{i \in S_b} p_{B_i} = \sum_{i \in S_b} b = |S_b|\cdot b \geq \frac{\opt}{O(\log \; m)},
\end{align*}
where the last inequality is by \eqref{eq_rev_Sb}. This proves the welfare guarantee for the first case. 

\item \emph{Most of the welfare in allocation $A$ comes from consumers in $\bar{E}$:}
$$
\sum_{i \in \bar{E}} v_i(A_i) \geq \frac{1}{2} \sum_{i\in N} v_{i}(A_{i}).
$$
For every $i\in \bar{E}$ place all items in $A_i$ in the same bundle. Let $T$ be the set of all remaining items (items that were allocated to consumers in $E$ or were not allocated at all). Add $T$ to the bundle of consumer $i^*$ where $i^* \in \arg\max_{i\in \bar{E}} \{ v_i(A_i \cup T)\}$. Denote the resulting bundles by $B_i$. Set the price of $B_i$ to be $p_{B_i} = v_i(B_i)$. We claim we have constructed a \pcbe{} with respect to the consumers in $\bar{E}$. To prove this we need to show that for every $i\in \bar{E}$, $B_i$ is in the demand set of $i$. 

Fix a consumer $i\in \bar{E}$. Claim \ref{cl_greedy_approx} guarantees that for any $i'\in \bar{E}\setminus\{i^*\}$, $v_i(B_{i'}) \le \sum_{j\in B_{i'}} v_i(j) \leq \sum_{j\in B_{i'}} v_{i'}(j) = v_{i'}(B_{i'})=p_{B_{i'}}$, so $i$ gains no profit from bundle $B_{i'}$. Furthermore, he gains no profit from $B_{i^*}$, since by definition $i^*$ is the value-maximizing consumer for this bundle and the price is set to be $i^*$'s value. 
Thus, we get a \pcbe{} and in particular:
\[ 
\sum_{i\in \bar{E}} p_{B_i} = \sum_{i \in \bar{E}} v_i(B_i) \geq \frac{1}{2} \sum_{i\in N} v_{i}(A_{i}) \geq \opt/8,  
\]
where the final inequality is by Claim \ref{cl_greedy_approx}. Using Lemma \ref{lem:fcwe-from-partial} we conclude that there is a \fcwe{} with welfare at least $\opt/8$.
\end{enumerate}

\end{proof}

\subsection{An Impossibility Result}

The next result shows that a constant-factor loss in welfare is inevidable for budget-additive valuations.

\begin{proposition}
\label{thm:ba-negative}
There exists a combinatorial market with budget-additive valuations such that every \fcwe{} has welfare that is at most a $\frac{4}{5}$-fraction of the optimal welfare.
\end{proposition}

\begin{proof}
Let $\epsilon,\delta$ be arbitrarily small positive constants and consider the budget-additive instance in Table \ref{tab:BA}. It is not hard to verify that the optimal integral allocation is obtained by giving items $a,b$ and $c$ to consumers $1,2$ and $3$ respectively, yielding $\opt = 5-\epsilon/2-\delta$. As in Section \ref{sec:two-buyers}, we will show there exists a better fractional solution to the configuration LP, implying there is no competitive equilibrium for this instance and hence bundling items is necessary to obtain a \fcwe{}. A straightforward calculation shows that in this case, the best solution is to give item $a$ to consumer $1$ (for a price of $2-\epsilon$), and the bundle $\{ b,c\}$ to consumer $2$ (for a price of $2$), yielding total welfare of $v_1(a)+v_2(\{b,c\}) = 2+2 = 4$. Thus, the approximation ratio of the \fcwe{} approaches $5/4$ as $\epsilon,\delta$ approach $0$.

To see that there exists a better fractional solution, consider the solution which assigns to consumer $1$ item $a$ with probability $1/2$ and bundle $\{b,c\}$ with probability $1/2$, to consumer $2$ item {b} with probability $1/2$ and item {c} with probability $1/2$, and to consumer 3 item $a$ with probability $1/2$. This is a feasible solution to the LP and gives an objective value of 
\[ 
\frac{1}{2} (v_1(a) + v_1(\{b,c\})) + \frac{1}{2}(v_2(b) + v_2(c)) + \frac{1}{2}v_3(a) = 
2+2+\frac{2-\epsilon}{2} = 5 - \frac{\epsilon}{2}, 
\]
which is greater than the optimal integral solution.
\begin{table}
\centering{
\caption{A lower bound for budget-additive consumers \label{tab:BA}}{%
\begin{tabular}{|l|c|c|c|c|c|}
\hline
                       	& a  				& b    		& c 		 		& Budget 		 \\
\hline
Consumer $1$        	&  $2$                      	&  $1$ 		& $1$	 		& 2			 \\
Consumer $2$  		&   $0$    			&   $2$        	& $2$  			& $2$   		 \\
Consumer $3$  		&   $2-\epsilon$     	& $0$		&   $1-\epsilon/2-\delta$      & $2-\epsilon$   \\

\hline
\end{tabular}}}
\end{table}
\end{proof}

\section{Revenue Guarantees for Matroid Rank Valuations}
\label{sec:revenue}

The first welfare theorem guarantees that whenever a competitive equilibrium exists it maximizes social welfare.
However, there is no guarantee on the \emph{revenue} in equilibrium -- indeed, there are markets in which every competitive equilibrium has $0$ revenue, even if all consumers have unit-demand valuations. This is the case, for example, in a market with $m$ identical items and $n<m$ unit-demand consumers: since supply exceeds demand, the equilibrium price of every item is $0$.

In this section we study revenue guarantees in two interesting subclasses of gross substitutes valuations (which, in particular, both include unit-demand valuations).
For both classes we show that a logarithmic fraction of the optimal welfare can be extracted in a competitive bundling equilibrium.
Note that this is the best fraction one can hope for, as is evident from the following example, which is a simplified version of the market that appears in the proof of Theorem \ref{thm:multi-tight}.
\begin{example}
Suppose there are $m=n$ items and consumer $i$ is unit-demand with value $1/i$ for any item. By applying similar arguments to the ones appearing in the proof of Theorem \ref{thm:multi-tight}, in every competitive bundling equilibrium, the total price for all the bundles is at most $1$, whereas the optimal welfare is $\ln n$ (obtained by assigning a single item to each consumer).
\end{example}
We note that this example was previously observed in \cite{FGL}.

Before turning to the statement of the revenue guarantees, let us recall
the formal definition of a matroid (see, e.g., \cite{Sch03}), and present the two subclasses considered in our theorem.
A set system $(M,\cI)$ (where $M$ is a set of items and $\cI$ is a family of independent subsets of $M$) is a matroid if (1) $\emptyset\in \cI$, (2) the family of independent sets $\cI$ is downward closed and (3) the \emph{exchange property} holds: if $S,T\in\cI$ and $|S|>|T|$, then there exists an element $j\in S\setminus T$ such that $T\cup\{j\}\in\cI$. One of the simplest matroids is a \emph{$k$-uniform} matroid, whose independent sets are simply all sets of size at most $k$.

The matroid \emph{rank function} $r:2^M \to\mathbb{Z^+}$ assigns every set $S\subseteq M$ the size of the largest independent set which is a subset of $S$.
The \emph{rank} of a matroid is the maximum rank over all sets. A valuation $v$ is a \emph{weighted matroid rank function} if there is a matroid $(M,\cI)$, and non-negative weights $\{w_{j}\}_{j\in M}$ such that for every item set $S\subseteq M$, $v(S)$ is the weight of the maximum-weight independent set which is a subset of $S$:
\[ v(S) = \max_{T \subseteq S,\; T\in \cI} \sum_{j\in T}{w_{j}}. \]

In this section we consider the following two settings:
\begin{itemize}
\item Every consumer has a weighted matroid rank function of a uniform matroid, possibly with different ranks and different weights.
\item Every consumer has a weighted matroid rank function of the same general matroid, possibly with different weights.
\end{itemize}
Note that the two settings are incomparable.
In the former, different consumers have different matroids, but all matroids are uniform, while in the latter, all consumers have the same matroid (possibly with different weights), but that matroid is not necessarily uniform.
With this we are ready to state our result.
\begin{theorem}
\label{thm:matroid_rev}
For every market with consumers whose valuations are weighted matroid rank functions, in either of the two settings above, there exists a \fcwe{} that extracts as revenue at least $\Omega(\frac {1} {\log m})$ of the optimal welfare.
\end{theorem}
The proof of the theorem appears in Appendix \ref{apx:revenue}.

\section{Optimal Welfare and Relation to Non-Linear Pricing Equilibrium}
\label{sec:LP}

In this section we discuss the relation between the solution concept of \fcwe{}, and that of \bikh{} (a.k.a.~second-order pricing equilibrium) due to \citet{BO02}. The motivation for their concept is similar to our motivation -- they develop the \emph{package assignment model}, in which trade can be of bundles rather than just individual items, and aim to extend the scope of \we{} to this model.%
\footnote{Their model also includes multiple sellers.} %
A \bikh{} is an allocation $(S_1,\dots,S_n)$ of items to consumers, and a vector of $2^m$ bundle prices $\vec p$, such that:
\begin{enumerate}
\item \textbf{Profit Maximization:} For every consumer $i$ and alternative set of items $T$, $v_i(S_i)-p_{S_i} \geq v_i(T)-p_T$.
\item \textbf{Revenue Maximization:} For every alternative allocation $(T_1,\dots,T_n)$, $\sum_i {p_{S_i}}\ge \sum_i {p_{T_i}}$.
\end{enumerate}
The main difference between this and our solution concept is the use of non-linear prices, whereas in a \fcwe{} the prices given the bundling are additive.

Bikhchandani and Ostroy design a linear program called $\CAPTWO$ and show that it characterizes \bika{}, i.e., such an equilibrium exists if and only if the optimal fractional solution to $\CAPTWO$ occurs at an integral solution, and if so the primal and dual variables correspond to the equilibrium allocation and prices. As a corollary they get that \bika{} -- when exist -- achieve optimal welfare.

We show here that existence of a \bikh{} implies the existence of a \fcwe{} with optimal welfare; this implication of \bika{} was not previously observed to our knowledge, and it can be viewed as a version of the \emph{second welfare theorem} for \fcwa{} -- if an integral optimal solution exists for $\CAPTWO$, then a \fcwe{} whose allocation corresponds to this solution also exists \cite[c.f.,][Theorem 11.15]{BN07}). One implication of the above is that for markets with superadditive%
\footnote{A valuation $v$ is superadditive if for every two sets of items $T,U$ it holds that $v(T)+v(U)\le(T\cup U)$.} %
consumers, there exists an efficient \fcwe{}, and moreover it can be reached via an ascending auction \citep{PU00,SY14}. We also identify market settings in which a \bikh{} does not exist, yet there exists a \fcwe{} with optimal welfare. We interpret this as indicating that there is a sense in which the solution concept of \fcwe{} dominates that of \bikh{}, in terms of the welfare it can achieve stably.  

We summarize the relation between the two solution concepts as follows:

\begin{proposition}
\label{pro:relation-to-second-order}
Consider a combinatorial market.
\begin{enumerate}[(a)]
\item The existence of a \bikh{} implies the existence of an efficient \fcwe{}. 
\item The existence of an efficient \fcwe{} does not always imply the existence of a \bikh{}. 
\end{enumerate}
\end{proposition}

\begin{proof}
By combining Lemma \ref{lem:CAP2} with Example \ref{ex:first} or Example \ref{ex:last}.
\end{proof}

The proof of part (a) of the proposition builds upon complementary slackness properties of the optimal integral solution to $\CAPTWO$ (see Lemma \ref{lem:CAP2} in Section \ref{apx:LP1}); the proof of part (b) is by showing an integrality gap for the same linear program (see examples in Section \ref{apx:LP2}). For completeness we state here the linear program $\CAPTWO$ and its dual. As usual $S$ denotes a set of items and $\cB$ denotes a bundling. The variables of the primal are $x_{i,S}$, which is the fraction of bundle $S$ allocated to consumer $i$, and $z_\cB$, which is the fraction of bundling $\cB$ determining the bundles that can be allocated. The dual variables are $\pi_i$, which is the payoff of consumer $i$, $\pi_0$, which is the revenue, and $p_S$, which is the (non-linear) price of bundle $S$. $\CAPTWO$ and its dual are:
\begin{eqnarray}
\max &\sum_{i,S} (x_{i,S}v_i(S))&\nonumber\\
\text{s.t.} & \sum_S x_{i,S} \le 1 & \forall i\in N \label{eq:one-bundle-each}\\
& \sum_i x_{i,S} \le \sum_{\cB:S\in\cB}z_\cB & \forall S\subseteq M \label{eq:bundle-in-bundling}\\
& \sum_\cB z_\cB \le 1 &\label{eq:single-bundling}\\
& x_{i,S},z_\cB\ge 0. \nonumber\\
& & \nonumber\\
\min & \pi_0 + \sum_i \pi_i &\nonumber\\
\text{s.t.} & \pi_i \ge v_i(S) - p_S & \forall i\in N,S\subseteq M \label{eq:cond1}\\
& \pi_0 \ge \sum_{S\in\cB} p_S & \forall \cB \label{eq:cond2}\\
& \pi_0,\pi_i,p_S \ge 0. &\nonumber
\end{eqnarray}


\subsection{\BIKH{} Implies Efficient \FCWE{}}
\label{apx:LP1} 

\begin{lemma}
\label{lem:CAP2}
If the optimal fractional solution to $\CAPTWO$ occurs at an integral solution with allocation $(S_1,\dots,S_n)$ over bundling $\cB$ and non-linear prices $\vec{p}$, then $(S_1,\dots,S_n)$ over $\cB$ with bundle prices $\vec{p'}$ according to $\vec p$ form an efficient \fcwe{}. 
\end{lemma}

\begin{proof}
We begin with the optimal integral solution above. Notice that by Conditions (\ref{eq:one-bundle-each}) and (\ref{eq:bundle-in-bundling}), $S_i\in \cB$ for every $i$. The lemma follows immediately from three claims:
\begin{enumerate}
\item If $S\in\cB$ is unallocated, $p_S=0$: By complementary slackness, for every $S\in\cB$ such that $p_S>0$, Condition (\ref{eq:bundle-in-bundling}) must hold with equality, so $S$ is allocated to some consumer $i$.
\item The allocation $(S_1,\dots,S_n)$ maximizes welfare: This follows from \cite{BO02}. 
\item For every consumer $i$ and bundle set $T$ over $\cB$, $v_i(S_i)-p'_{S_i} \ge v_i(T)-\sum_{B\in T} p'_B$.
\end{enumerate}

We now prove the third claim. The bundle prices $\vec{p'}$ are set according to $\vec p$, and so $p'_B=p_B$ for every $B\in\cB$. This means we need to show that $v_i(S_i)-p_{S_i} \ge v_i(T)-\sum_{B\in T} p_B$. By complementary slackness and Condition (\ref{eq:cond1}), 
$$
v_i(S_i)-p_{S_i} = \pi_i \ge v_i(T)-p_T,
$$ 
where slightly abusing notation, $p_T$ is the price of the set $\bigcup_{B\in T}B$. It is left to show that $p_T \le \sum_{B\in T} p_B$. By complementary slackness and Condition (\ref{eq:cond2}), this subadditivity of prices indeed holds:
$$
\sum_{B\in\cB}p_B = \sum_{B\in\cB\setminus T}p_B + \sum_{B\in T}p_B = \pi_0 \ge \sum_{B\in\cB\setminus T}p_B + p_T,
$$
and so $\sum_{B\in T} p_B \ge p_T$ as required.

Note that strictly speaking we require a \fcwe{} to have $\bigcup_i S_i = M$; by the first claim above we can simply allocate any unallocated $S\in\cB$ with price $p_S=0$ to an arbitrary consumer, and the third claim will still hold. This completes the proof of the lemma.
\end{proof}

\subsection{Efficient \FCWE{} Does Not Imply \BIKH{}}
\label{apx:LP2} 

The examples below show markets in which an efficient \fcwe{} exists but a \bikh{} does not. To prove the latter we provide for each example a certificate showing an integrality gap for $\CAPTWO$.

\begin{example}[Two items]
\label{ex:first}
Consider a combinatorial market with two consumers and two items. Consumer 1 has value of 8 for both items together and 0 otherwise. Consumer 2 is unit-demand and values each item at 7. An efficient \fcwe{} bundles both items together and allocates the bundle to consumer 1 for a price between 7 and 8. A certificate showing there is no \bikh{} is the following fractional solution to $\CAPTWO$ that achieves 11 instead of 8: $x_{1,ab} = x_{2,a} = x_{2,b} = 0.5$, and $z_{ab} = z_{a,b} = 0.5$. 
\end{example}

\begin{example}[Submodular valuations]
\label{ex:last}
Consider a combinatorial market with two consumers and four items $a,b,c,d$. Consumer 1's submodular valuation: All items valued at 1, the pairs $ab, cd$ are valued at 2, all other pairs $ac, ad, bc, bd$ are valued at 1.5, all sets of three or four items valued at 2.
Consumer 2's submodular valuation: All items valued at 1, the pairs $ad, bc$ are valued at 2, all other pairs $ab, ac, bd, cd$ are valued at 1.5, all sets of three or four items valued at 2.
An efficient \fcwe{} bundles $ab$ together and $cd$ together, allocates the bundle $ab$ to consumer 1 for price 1.5, and allocates $cd$ to consumer 2 for price 1.5.
A certificate showing there is no \bikh{} is the following fractional solution to $\CAPTWO$ that achieves 4 instead of 3.5: $x_{1,ab} = x_{1,cd} = x_{2,ad} = x_{2,bc} = 0.5$, and $z_{ab,cd}=z_{ad,bc}=0.5$.
\end{example}

\section{Welfare Maximization: Does Randomness Help?}
\label{apx:randomness}

Given the negative result in Theorem \ref{thm:multi-tight}, stating that \fcwa{} cannot guarantee more than a logarithmic fraction of the optimal welfare in the worst case, a natural question is whether introducing randomness in the form of lotteries can help circumvent this lower bound. 

Randomness in the context of \we{} dates back to \cite{PT84}. \citet{Gar95} studies \emph{lottery equilibria}, and while his paper does not involve bundling, we can build upon its notions. In particular, a market with lotteries is one in which the allocation to every consumer is a distribution over bundles -- that is, a lottery ticket -- and the allocation of lottery tickets to consumers must correspond to some joint probability distribution over item allocations. Consumers aim to maximize their expected profit.

We now define a \fcwe{} with lotteries: This is a \fcwe{} over bundles of lottery tickets, whose allocation must correspond to a distribution over item allocations, with an additional requirement of \emph{ex post} market clearance. I.e., we require that after the lotteries are run, the market clears from all the items. The reason for this is discussed in the Introduction -- to maintain the concept of an equilibrium, producers must not have the opportunity to lower prices and sell more. As further discussed there, such equilibria are not applicable to all markets, and require an assumption of trade restriction among the consumers after the lotteries are run.%
\footnote{Note that if this assumption is dropped, then regardless of the use of lotteries, the ex post allocation should form a \we{} or \fcwe{}, and so lotteries fail to enrich the solution concept -- c.f., \cite{BO02}.}

\begin{definition}
\label{def:CBE-lotteries}
A \fcwe{} \emph{with lotteries} is a \fcwe{} in a market of lottery tickets, whose allocation corresponds to a distribution over item allocations, each of which clears the market of items.
\end{definition}

We can now state our main result for this section -- an impossibility result with lotteries.

\begin{theorem}
\label{thm:multi-lotteries}
There exists a multi-unit setting where $n=m$ and valuations are subadditive, such that every \fcwe{} with lotteries has welfare which is a $(1/\Omega(\log m))$-factor of $\opt$.
\end{theorem}
 
The proof is by an interesting connection between our market setting and between optimal mechanism design in a single-parameter auction setting.

\begin{proof}[of Theorem \ref{thm:multi-lotteries}]
Recall the setting introduced in the proof of Theorem \ref{thm:multi-tight}. We will refer to this setting here as the \emph{equal-price market}. Recall from the proof of Theorem \ref{thm:multi-tight} that the logarithmic lower bound follows from the fact that in any \fcwe{}, all units are allocated to consumer 1. We now prove that in any \fcwe{} with lotteries, ex post all the units are allocated to consumer 1.

Consider a \fcwe{} with lotteries. Let $\alpha$ denote the probability with which consumer 1 receives all units. Then due to ex post market clearance, $1-\alpha$ is the probability that at least one consumer $2,\dots,n$ receives at least one unit. We argue that it is sufficient to show that the total payment of the unit-demand consumers cannot exceed $1-\alpha$. Indeed, assume this holds and consider consumer 1. Assume for contradiction that consumer 1 is allocated a subset of lottery tickets that does not guarantee him all units ex post, i.e., $\alpha<1$. His expected value is then $(1-\alpha)(1+\epsilon) + \alpha(2+2\epsilon) = (1+\alpha)(1+\epsilon)$. If he were to buy all remaining lottery tickets instead of the unit-demand consumers, the increase in his expected profit would be at least $(1-\alpha)(1+\epsilon) - (1-\alpha) > 0$. This is by our assumption that the total payment of the unit-demand consumers is $\le 1-\alpha$. We thus have a contradiction -- if there is a chance that bidder 1 does not receive all units ex post, he can strictly improve his expected profit by buying all lottery tickets. 

In the remainder of the proof we show the following claim:
 
\begin{claim}
\label{cla:total-pay}
The total payment of the unit-demand consumers is upper-bounded by $1-\alpha$. 
\end{claim}

To prove this claim, we reduce our setup to a Bayesian auction setting, with a single bidder and a single item for sale. 
We will then show that a \fcwe{} with lotteries in the equal-price market corresponds to a truthful and individually rational (IR) randomized mechanism in the auction setting, such that the expected revenue of the auction mechanism is at least the total payment of the unit-demand consumers in the equilibrium. 
We can then apply Myerson's theorem to show that the optimal expected revenue in the auction setting is at most $1-\alpha$, implying that the total payment of the unit-demand consumers cannot exceed $1-\alpha$. 

The \emph{equal-revenue} auction setting is as follows: There is a single item for sale and a single bidder, whose value is uniformly distributed over the $n-1$ values in the discrete range $\{1/2,\dots,1/n\}$. Denote the discrete value distribution by $F$. 

\begin{claim}
\label{cla:reduction-to-mech}
For every \fcwe{} with lotteries in the equal-price market, in which consumer 1 receives all units with probability $\alpha$, there is a truthful IR randomized mechanism in the equal-revenue auction, such that:
\begin{itemize}
\item For every bid, the probability that the mechanism allocates the item to the bidder is at most $1-\alpha$;
\item The expected revenue of the mechanism is at least a $(1/(n-1))$-fraction of the total payment of the unit-demand consumers in the market.
\end{itemize}
\end{claim}

\begin{proof}[of Claim \ref{cla:reduction-to-mech}]
Consider a \fcwe{} with lotteries in the equal-price market. We introduce the following notation: for every $i\in\{2,\dots,n\}$, let $(x_i,p_i)$ denote the total probability with which consumer $i$ is allocated at least one unit in the \fcwe{} with lotteries, and let $p_i$ be the price that consumer $i$ is charged. Consumer $i$'s profit is $(x_i/i) - p_i$ (using that we're in a multi-unit setting with unit-demand consumers). Note that $x_i\le 1-\alpha$ by the assumption that consumer 1 receives all units with probability $\alpha$. 

By the well-known taxation principle of mechanism design, every truthful and IR mechanism in the equal-revenue auction can be described as a menu of $(x_i,p_i)$ pairs, where $x_i$ is the allocation probability of the item, and $p_i$ is the price. The bidder chooses among the menu options according to his value for the item, picking the most profitable one for him. Also, every such price menu corresponds to an IR mechanism, as long as for every value in the range of $F$, there is at least one menu option that results in non-negative expected profit for the bidder. 

We can thus take the $(x_i,p_i)$ pairs of the unit-demand consumers in the equal-price market, and view them as an IR mechanism in the equal-revenue auction. The mechanism is indeed IR because for every value $1/i$ in the support of $F$, the pair $(x_i,p_i)$ results in non-negative expected profit $(x_i/i) - p_i$ for the bidder, otherwise we would have a violation of profit maximization for consumer $i$ in the \fcwe{} with lotteries (Definition \ref{def:CBE-lotteries}). 

The resulting mechanism is also truthful, in the sense that when the bidder's value is $1/i$, a profit-maximizing menu option for him is $(x_i,p_i)$. This is again since we have ex ante profit maximization in the \fcwe{} with lotteries, i.e., consumer $i$ is allocated one of his demand sets. Now we see that the two properties in Claim \ref{cla:reduction-to-mech} hold -- for every menu option the allocation probability is $x_i\le 1-\alpha$, and the expected revenue is the average menu price, i.e., $1/(n-1)$ times the total payment of the unit-demand consumers.
\end{proof}

\begin{claim}
 \label{cla:revenue-bound}
Consider a truthful IR mechanism in the equal-revenue auction that allocates the item to the bidder with probability at most $1-\alpha$ for every bid. Then its expected revenue is at most $(1-\alpha)/(n-1)$.
\end{claim}   

\begin{proof}[of Claim \ref{cla:revenue-bound}]
Consider a truthful IR mechanism as defined. The equal-revenue auction is a single-parameter setting, and $F$ is a (discrete) regular distribution.%
\footnote{The virtual value corresponding to a bid $1/i$ is $1/i(i-1)$. For more on discrete regular distributions see, e.g., \cite{CHRR06}).
} %
Thus by Myerson's theory of optimal mechanisms, the expected revenue of the mechanism is equal to its expected virtual surplus, and the optimal expected revenue can be achieved by a deterministic mechanism \cite{M81}. For the equal-revenue auction, any deterministic mechanism is simply a reserve price, and it is not hard to see that every reserve price extracts revenue of $\le 1/(n-1)$. The optimal expected virtual surplus is thus $\le 1/(n-1)$. Observe that the optimal expected virtual surplus is achieved by allocating the item to the bidder with probability 1 if and only if his bid yields a positive virtual value. Since we require the mechanism to allocate the item with probability $\le 1-\alpha$ for every bid, the maximum virtual surplus it can achieve is $\le (1-\alpha)/(n-1)$.
\end{proof}

This concludes the proof of Claim \ref{cla:total-pay} and Theorem \ref{thm:multi-lotteries}.
\end{proof}

\bibliographystyle{acmsmall}
\bibliography{bundling-bib}

\appendix
\section{Missing Proof: Two Consumers}
\label{apx:two-buyers}

\begin{proof}[of Proposition \ref{prop_two_consumers_ub}]
We first claim that we can assume without loss of generality that there are only two items in the market (i.e., $m=2$). This is true since we can always bundle the $m$ items into (at most) two bundles, according to the optimal solution.
That is, each of the two bundles would consist of the items allocated to one of the consumers in the optimal allocation, and the new values would be the values over these bundles. Clearly, this does not affect the value of $\opt$, and any \fcwe{} that is produced from these bundles is a valid \fcwe{} in the original market.

Therefore, it is sufficient to prove the theorem for an instance with two consumers, $1$ and $2$, and two items, $a$ and $b$. Let $\xa,\xb,\xab$ denote consumer $1$'s values for $a$, $b$, and the bundle $\{a,b\}$, respectively, and let $\ya,\yb,\yab$ denote the analogous values of consumer $2$.
Without loss of generality, we may assume that
\[\opt = \xa + \yb, \] as otherwise the optimal solution is supported by a \fcwe{}. Also assume (without loss of generality) that $\xa \geq \yb$. In addition, we may assume that one of the consumers has a subadditive valuation, and the other has a superadditive valuation. Indeed, for the case of two items, the class of subadditive valuations coincides with gross substitutes valuations, which always admits an item-price \we{}, and therefore obtains the optimal welfare.
Likewise, if both consumers are superadditive, then there always exists a \fcwe{} that obtains the optimal welfare.

Clearly, if either $\xab \geq \frac{2}{3}\cdot \opt$ or $\yab \geq \frac{2}{3}\cdot \opt$, then there exists a \fcwe{} that obtains at least $\frac{2}{3}$ of the welfare by bundling the two items together. Therefore, in the remainder of the proof we assume that both values are strictly smaller than $\frac{2}{3}\cdot \opt$. Substituting $\opt=\xa+\yb$, we get
\begin{align}
\xab , \yab < \frac{2}{3}(\xa + \yb). \label{eq_grand_bundle_small}
\end{align}

By monotonicity (i.e., since $\xa \leq \xab$), it also holds that
\begin{align} \label{eq_y_b_big}
\yb \geq \frac{1}{2}\xa.
\end{align}

We now distinguish between two cases, and establish the existence of an item-priced \we{} for both of them. 
Recall that, by the first welfare theorem, if a \we{} exists, then it obtains the optimal welfare.

\vspace{0.1in}\noindent \bf Case 1: \rm Consumer $1$ is superadditive, consumer $2$ is subadditive  ($\xab \geq \xa + \xb$, \; $\yab \leq \ya + \yb$).   \\Combining the superadditivity of consumer $1$ with Eq. \eqref{eq_grand_bundle_small} implies that either $\xa < \frac{1}{3}\xa + \frac{1}{3}\yb $ or $ \xb < \frac{1}{3}\xa + \frac{1}{3}\yb$. Since the first inequality contradicts the assumption that $\xa \geq \yb$, it follows that
\begin{align} \label{eq_x_b_small}
\xb < \frac{1}{3}\xa + \frac{1}{3}\yb.
\end{align}
Let $p_a=\xa$ and $p_b= \frac{1}{3}\xa + \frac{1}{3}\yb.$
We claim that allocating item $a$ to consumer $1$ for a price of $p_a$ and item $b$ to consumer $2$ for a price of $p_b$ forms a \we{}.

Consumer $1$ obtains $0$ utility from his allocation, so we need to show that no other allocation gives him a strictly positive utility. The utility consumer $1$ derives from item $b$ is $u_1(b) = \xb - p_b = \xb-(\frac{1}{3}\xa + \frac{1}{3}\yb)$, which is negative by  Eq. \eqref{eq_x_b_small}.
For the bundle $\{a,b\}$, it holds that
\begin{align} \label{eq_u_1_ab}
u_1(\{a,b\}) = \xab - (p_a+p_b) < \frac{2}{3}\xa + \frac{2}{3}\yb - \xa - \frac{1}{3}\xa - \frac{1}{3}\yb =
\frac{1}{3}\yb - \frac{2}{3}\xa \leq 0,
\end{align}
where the first inequality follows by Eq. \eqref{eq_grand_bundle_small}, and the last inequality follows by the assumption that $\xa \geq \yb$.

For consumer $2$, it holds that $u_2(b) = \yb-p_b = \yb - (\frac{1}{3}\xa + \frac{1}{3}\yb)$, which is non-negative by Eq.
\eqref{eq_y_b_big}.
We next show that $u_2(a)<u_2(b)$.
It holds that
\[  u_2(a) - u_2(b) = \ya - p_a - (\yb - p_b) \leq \yab - p_a - \yb + p_b <
\frac{2}{3}\xa + \frac{2}{3}\yb - \xa - \yb + \frac{1}{3}\xa + \frac{1}{3}\yb = 0,\]
where the first inequality follows by monotonicity (in particular, $\ya \leq \yab$), and the second inequality follows by substituting the payments.
Finally, it holds that $u_2(\{a,b\}) < 0$, by the same calculation as in \eqref{eq_u_1_ab} (replacing $\xab$ with $\yab$).
We conclude that the proposed allocation and payments constitutes a \we{}, as promised.

\vspace{0.1in}\noindent \bf Case 2: \rm Consumer $1$ is sub-additive, consumer $2$ is superadditive  ($\xab \leq \xa + \xb$, \; $\yab \geq \ya + \yb$).   \\

The existence of a \we{} in this case is established in a similar way to case 1.
Combining the superadditivity of consumer $2$ with assumption \eqref{eq_grand_bundle_small} implies that either $\yb < \frac{1}{3}\xa + \frac{1}{3}\yb$ or $\ya < \frac{1}{3}\xa + \frac{1}{3}\yb$.
Since the former case contradicts Eq. \eqref{eq_y_b_big}, it follows that
\begin{align} \label{eq_y_a_big}
\ya < \frac{1}{3}\xa + \frac{1}{3}\yb.
\end{align}

Let $p_a=\frac{1}{3}\xa + \frac{1}{3}\yb$ and $p_b= \yb.$
We claim that allocating item $a$ to consumer $1$ for a price of $p_a$ and item $b$ to consumer $2$ for a price of $p_b$ forms a \we{}.
For consumer $2$, $u_2(b)=\yb-p_b=0$. We show that consumer 2's valuation for any other allocation is non-positive.
For item $a$, $u_2(a) = \ya - p_a = \ya -  \left(\frac{1}{3}\xa + \frac{1}{3}\yb \right) < 0$, where the last inequality follows by Eq. \eqref{eq_y_a_big}.
For the bundle $\{a,b\}$, it holds that
\[ u_2(\{a,b\}) = \yab - (p_a+p_b) < \frac{2}{3}\xa + \frac{2}{3}\yb - \frac{1}{3}\xa - \frac{1}{3}\yb - \yb  = \frac{1}{3}\left(\xa-2\yb \right) \leq 0, \]
where the first inequality follows by Eq. \eqref{eq_grand_bundle_small} and the last one follows by Eq. \eqref{eq_y_b_big}.

It remains to show that the same holds with respect to consumer $1$.
$u_1(a)=\xa-p_a = \xa-(\frac{1}{3}\xa + \frac{1}{3}\yb)= \frac{2}{3}\xa - \frac{1}{3}\yb$, which is non-negative by the assumption that $\xa \geq \yb$.
For the bundle $\{a,b\}$, the same calculations as for consumer $2$ (replacing $\yab$ with $\xab$) shows that $u_1(\{a,b\}) < 0$.
Finally, we show that $u_1(b)<u_1(a)$.
Indeed,
\[ u_1(b) - u_1(a) = \xb - p_b - (\xa - p_a) \leq
\xab - p_b - \xa + p_a < \frac{2}{3}\xa + \frac{2}{3}\yb - \yb -\xa + \frac{1}{3}\xa + \frac{1}{3}\yb = 0, \]
where the first inequality follows by monotonicity and the last one follows by Eq. \eqref{eq_grand_bundle_small}.

\vspace{0.1in} We conclude that it is either the case that one of the consumer derives at least $\frac{2}{3}$ of the optimal welfare from the grand bundle, or there exists an item-price \we{}, in which case the optimal welfare is obtained. Thus, the assertion of the theorem is established.
\end{proof}

\section{Missing Proofs: Preperations}
\label{apx:prep}

\begin{proof}[of Lemma \ref{lem:fcwe-from-partial}]
Fix some $\epsilon>0$. For every consumer $i\in N'$ for which $S_i$ is not empty, define a new valuation $v^\epsilon_i$ that is identical to $v_i$ except for a shift of $\epsilon$ in the value of $S_i$, i.e., $v^\epsilon_i(S_i)=v_i(S_i)+\epsilon$ (the new valuation may no longer be monotone). For every other consumer $i$ simply set $v^{\epsilon}_i=v_i$. Observe that the \pcbe{} is still a \pcbe{} with respect to the $v^\epsilon_i$'s. Now apply Theorem \ref{lemma-FGL} to get a bundling $\mathcal B'$, allocation $(S'_1,\dots,S'_n)$ and prices $\vec{p'}$. We show that since we started with a \pcbe{}, these bundling, allocation and prices form a \fcwe{} with respect to the $v^\epsilon_i$'s, that is, all bundles in $\mathcal B$ are allocated. By the latter fact and by properties (\ref{item-FGL-prices-up}) and (\ref{item-FGL-profit-maximization}) of Theorem \ref{lemma-FGL}, $\sum_{i\in N} v_i(S'_i)\geq \sum_{B\in\mathcal{B}} p_B$.

To show that we get a \fcwe{}, since property (\ref{item-FGL-profit-maximization}) of Theorem \ref{lemma-FGL} is guaranteed, the only missing component is to show market clearance, i.e., that $\cup_i S_i=M$. Suppose towards a contradiction that there is a bundle $B\in \mathcal B$ that was not allocated. Let $i$ be the consumer that was allocated that bundle in the \pcbe{}. Observe that under the prices of the \pcbe{}, $B$ is the most profitable bundle of $i$. Now since $B$ is unallocated, its price remaines the same by property (\ref{item-FGL-unallocated}) of Theorem \ref{lemma-FGL}, while the prices of the other bundles can only increase by properties (\ref{item-FGL-prices-up}) and (\ref{item-FGL-unallocated}). Thus $B$ is \emph{the} most preferred bundle for $i$ (with valuation $v_i$ it is only a most preferred bundle), and by property (\ref{item-FGL-profit-maximization}) consumer $i$ must be allocated this bundle.

We would now like to show the existence of a \fcwe{} with respect to the $v_i$'s and not just with respect to the $v_i^{\epsilon}$'s. When taking $\epsilon$ to $0$, we get an infinite sequence of allocations and prices. Since the number of allocations is finite and since all prices are bounded between $0$ and $\max\{\max_iv_i(M), \max_B(p_B)\}$, there exists a subsequence in which one allocation $\tilde S$ repeats and the prices converge to a price vector $\tilde p$. Note that it still holds that $\sum_{i\in N} v_i(\tilde S_i)\geq \sum_{B\in\mathcal{B}} p_B$.

To finish the proof we now claim that this allocation $\tilde S$ and prices $\tilde p$ are a \fcwe{} with respect to the $v_i$'s. Observe that for every $\epsilon$ in the converging subsequence, if consumer $i$ receives $\tilde S_i$, 
then $\tilde S_i$ is the \emph{unique} bundle that maximizes his profit; otherwise, for smaller values of $\epsilon$, $\tilde S_i$ is no longer the most profitable bundle for $i$ in contradiction to the assumption that we have a \fcwe{} for the $v_i^{\epsilon}$'s. Since this is true for every $\epsilon>0$, for $\epsilon=0$ we get that $\tilde S_i$ is one of the most profitable bundles for $i$, which is enough to prove that $\tilde S$ and $\tilde p$ form a \fcwe{} with respect to the bundling $\mathcal B$.
\end{proof}

\begin{proof}[of Lemma \ref{lem:promising}]
We will show that after applying Theorem \ref{lemma-FGL} to the \promising priced bundling, all bundles in $\mathcal B$ are allocated. Assume towards contradiction there is a bundle $B\in \mathcal B$ that is not allocated. Since $B$ is unallocated its price remaines unchanged by property (\ref{item-FGL-unallocated}) of Theorem \ref{lemma-FGL}. But there are at least $|\mathcal{B}|$ consumers for which $B$ is profitable, hence for property (\ref{item-FGL-profit-maximization}) to hold, each of these consumers must be allocated an alternative bundle (otherwise their profit would be $0$). However, there are only $|\mathcal B|-1$ bundles except for $B$, a contradiction.
\end{proof}

\begin{proof}[of Lemma \ref{lem:logarithmic}]
Let $W = \sum_{i\in N}{v_i(S_i)}$ be the welfare of allocation $S$, and recall $\mu=\min\{n,m\}$. We define an auxiliary allocation $S''$ to be almost the same allocation as $S$, but without parts contributing little to the total welfare: for every $i\in N$, if $v_i(S_i) < W/2\mu$ then set $S''_i=\emptyset$, otherwise set $S''_i=S_i$. The welfare does not decrease by much:
$$
\sum_{i\in N}{v_i(S''_i)} \ge \frac{1}{2}W = \frac{1}{2} \sum_{i\in N}{v_i(S_i)}.
$$ 

We now partition consumers into bins as follows: put consumer $i$ into a bin corresponding to value $v$ if $v_i(S''_i) \in [v,2v)$. Notice that by construction of $S''$ it holds that $\forall i : v_i(S''_i) \in [W/2\mu,W]$, and so  there are at most $\log(2\mu)+1=\log(\mu)+2$ bins. Denote the set of consumers in the bin corresponding to the value $v$ by $N_v$. There must be some value $v^*$ such that the total value of consumers in $N_{v^*}$ is at least the overall total value $\sum_{i\in N}{v_i(S''_i)}$ divided by $\log(\mu)+2$. The proof is concluded by setting $S'_i=S''_i$ for every $i\in N_{v^*}$, and $S'_i=\emptyset$ otherwise. Observe that the construction requires $\poly(m,n)$ time and uses only value queries.
\end{proof}

\section{Missing Proofs: Revenue}
\label{apx:revenue}

\subsection{Proof of Theorem \ref{thm:matroid_rev}}

We divide the proof of Theorem \ref{thm:matroid_rev} into two parts, proving first the bound for uniform-rank functions and then for the case where all valuations have the same matroid. We begin with a lemma that will be useful in both parts.
For this lemma we use the following notion:
A competitive equilibrium with reserve $q$ is a competitive equilibrium in which the price of every item is at least $q$.
Specifically, while in a standard competitive equilibrium it is required that an item that is not sold has price $0$, here an item that is not sold has price $q$.

\begin{lemma}
\label{lem:WE-high-rev-extra}
In every combinatorial market with valuations that are weighted matroid rank functions,
there exists some price $q$ such that there exists a competitive equilibrium with reserve $q$ that extracts revenue of at least $\Omega(\frac {1 }{\log m})$ of the optimal welfare.
\end{lemma}

\begin{proof}
Following a similar argument to the one developed in Lemma \ref{lem:logarithmic}, there exists a value $v$ and an allocation $O'=(O'_1,\ldots, O'_n)$ such that for every item $j \in O'_i$, $v_i(\{j\}) \in [v,2v)$, and the welfare of $O'$ is a logarithmic approximation to the welfare of the optimal allocation.
Let $M_v = \bigcup_{i}O'_i$ be the items allocated in $O'$.
Then, the social welfare obtained in $O'$ is at most $2vM_v$, therefore we get
\begin{equation}
\label{eq:lb-opt}
2vM_v \geq OPT/O(\log m).
\end{equation}

Consider a modified market with an additional consumer $n+1$ who values every item at exactly $v$ (i.e., additive valuation).
Since an additive valuation satisfies the gross substitutes condition, this new market admits a competitive equilibrium.
Note that every item that is allocated to one of the original consumers (i.e., any consumer other than consumer $n+1$) has a price of at least $v$ (otherwise consumer $n+1$ is better off purchasing this item). Similarly, every item that is allocated to consumer $n+1$ has a price of exactly $v$.
Therefore, a competitive equilibrium in this modified market corresponds to a competitive equilibrium with reserve $v$ in the original market.
Among all competitive equilibria with reserve $v$, consider one that minimizes the number of unallocated items,
and let $A=(A_1, \ldots, A_n)$ denote the allocation in this equilibrium.

Since the price of every allocated item is at least $v$, due to Equation (\ref{eq:lb-opt}) it is sufficient to show that at least $M_v$ items are allocated in $A$.
Let $i$ be a consumer for whom $|A_i| < |O'_i|$.
By the exchange property of a matroid function, there exist $|O'_i| - |A_i|$ items that can be added to $A_i$, and for every such item $j \in O'_i \setminus A_i$, it holds that $v_i(\{j\}) \geq v$.
We claim that these items must be allocated in $A$.
Indeed, if not, then the price of $j$ is $v$, and since $v_i(\{j\}) \geq v$ it contradicts the minimality of the number of unallocated items in $A$.
The assertion of the lemma follows.
\end{proof}

The last lemma essentially proves the existence of a competitive equilibrium with respect to a \emph{subset} of the items that extracts the desired revenue.
The challenge is, therefore, to reallocate the unallocated items among the original consumers in a way that preserves the profit maximization property.
Note that we need not extract any revenue from these items, since the revenue is sufficiently high even without selling these items.

\subsubsection{Proof for Uniform Matroids}

Lemma \ref{lem:WE-high-rev-extra} establishes the existence of a competitive equilibrium with reserve $v$ that guarantees
a logarithmic fraction of the optimal welfare in revenue. Let $W$ and $p$ denote the allocation and prices of such an equilibrium, and let $M_W \subseteq M$ denote the allocated items in $W$.
In addition, we denote by $E$ (respectively, $\bar{E}$) the consumers who exhaust (resp., do not exhaust) their ranks.
i.e., $E= \{ i \in N \; \mid  \; |W_i| = k_i \}$.
Finally, let $M_E$ and $M_{\bar{E}}$ denote the items allocated to consumers in $E$ and $\bar{E}$, respectively.

We distinguish between two cases.

\vspace{0.1in} \noindent Case (a): Most of the revenue of $W$ comes from consumers in $E$; i.e.,
$\sum_{j\in M_E }{p(j)} \geq \frac{1}{2}\cdot \sum_{j\in M_W}{p(j)}$.

\vspace{0.1in} \noindent Consider the market that contains consumers in $E$ and the items $M_E$.
Observe that this market is supported by a competitive equilibrium with reserve $v$ in which all consumers exhaust their ranks (simply allocate the bundle $W_i$ to every consumer $i \in E$ at a price of $p(j)$).
Consider next the market that contains the consumers in $E$, but all items $M$.
We claim that this extended market also admits a competitive equilibrium with reserve $v$ in which all consumers exhaust their ranks. 
This is established by the following claim.

\begin{claim}
Consider a market with consumers with valuations that are weighted rank functions of a uniform matroid, and a set of items $M'$. Suppose this market admits a competitive equilibrium with reserve $v$ in which all consumers exhaust their ranks. 
Then, the market that contains the same set of consumers, and a set of items $M$ such that $M' \subseteq M$, also admits a competitive equilibrium with reserve $v$, in which all consumers exhaust their ranks.
\end{claim}

\begin{proof}
Let $(O_1,\ldots, O_n)$ be the allocation in the competitive equilibrium in the original market and $(O'_1,\ldots, O'_n)$ be the competitive equilibrium after adding items. Recall that by the first welfare theorem both allocations maximize the welfare in the corresponding markets.

Now observe that the optimal allocation if all valuations are rank functions of a weighted uniform matroid, the optimal solution can be computed by finding the maximum weighted matching in a bipartite graph where in one side every vertex corresponds to a different item, and in the other side there are $k_i$ vertices for each consumer $i$ with valuation of rank $k_i$. The weight of an edge between a vertex that corresponds to some item $j$ and each of the vertices of consumer $i$ has weight $v_i(\{j\})$.

Now, since one market is a subset of the other, the matching that corresponds to $(O_1,\ldots, O_n)$ is a valid matching in the bipartite graph that corresponds to the extended market. In particular, since $(O'_1,\ldots, O'_n)$ is optimal for the extended market, there is a sequence of reallocation that corresponds to augmenting paths that leads from $(O_1,\ldots, O_n)$ to an optimal allocation with value equal to the value of $(O'_1,\ldots, O'_n)$. Notice that by definition of augmenting paths if all the vertices in the side that corresponds to the players where originally matched all will continue to be matched after each augmentation. This implies that in competitive equilibrium with a minimal price $v$ all consumers still exhaust their ranks, as needed.
\end{proof}

We conclude that the market with consumers in $E$ and all items $M$ admits a competitive equilibrium with reserve $v$ in which all consumers exhaust their ranks.
Let $(W',p')$ be such an equilibrium, let $T=M \setminus M_E$ be the set of items that are not allocated in this equilibrium, and let $i_{max} \in \arg\max_{i\in E} k_i$ be a consumer with maximum rank among the consumers in $E$.
We are now ready to construct a \Pcbe{} with sufficiently high revenue.
Bundle all items in $W'_{i_{max}}\cup T$ together, and allocate it to consumer $i_{max}$ for a price of $\Sigma_{j\in W'_{i_{max}}}p'(j)$.
For every other consumer $i \neq i_{max}$, bundle all items in $W'_{i}$ together and allocate it to consumer $i$ for a price of $\Sigma_{j\in W'_{i}}p'(j)$.

We next show that this is a \Pcbe{} with respect to consumers in $E$.
For $i_{max}$, it is clearly the case that $W'_{i_{max}}\cup T$ is the most profitable bundle, since its value only increased and its price remained unchanged. Consider next a consumer $i \in E$ other than $i_{max}$.
Clearly, $i$ prefers his bundle to every other bundle of consumers other than $i_{max}$.
It remains to show that $i$ prefers his own bundle to receiving $W'_{i_{max}}\cup T$ for a price of $\Sigma_{j\in W'_{i_{max}}}p'(j)$.
By the maximality of $i_{max}$, consumer $i_{max}$ originally consumed at least $k_i$ items, each at a price at least $v$. 
In contrast, the price of every unallocated item was at most $v$. 
Therefore, the price of every bundle form $W'_{i_{max}}\cup T$ of size $k_i$ or less has either increased or remained unchanged.
The price of his own allocation, however, clearly remained unchanged.
Therefore, $W'_i$ is still the most profitable bundle for consumer $i$, as desired.
It follows that the above allocation and prices form a \Pcbe{} with revenue at least $OPT/O(\log m)$.
The proof is then concluded by invoking Lemma \ref{lem:fcwe-from-partial} to reach a \fcwe with the aforementioned revenue guarantees.

\vspace{0.1in} \noindent Case (b): Most of the revenue of $W$ comes from consumers in $\bar{E}$; i.e.,
$\sum_{j\in M_{\bar{E}} }{p(j)} \geq \frac{1}{2}\cdot \sum_{j\in M_W}{p(j)}$.

\vspace{0.1in} \noindent
Let $T=M \setminus M_{\bar{E}}$ be the set of items that are not allocated to a consumer in $\bar{E}$,
and let $i_{max} = \arg\max_{i\in \bar{E}}v_i(W_i \cup T)$.
We next construct a \Pcbe{} with sufficiently high revenue.
For every consumer $i \in \bar{E}$ except for $i_{max}$, bundle the items in $W_i$ and allocate this bundle to consumer $i$ for a price of $\Sigma_{j\in W_{i}}p(j)$.
Now, let $W'_{i_{max}} = W_{i_{max}} \cup T$ be a bundle that contains the items in $W_{i_{max}}$ and in $T$.
Allocate the bundle $W'_{i_{max}}$ to consumer $i_{max}$ for a price of $v_{i_{max}}(W'_{i_{max}})$.
We claim that this is a \Pcbe{} with respect to consumers in $\bar{E}$.
consider first consumer $i_{max}$.
The utility he derives from this outcome is $0$.
We claim that this is the most profitable outcome for $i_{max}$.
Indeed, since $|W_{i_{max}}| < k_{i_{max}}$, the utility that $i_{max}$ derives from every item $j \not\in W_{i_{max}}$ under price $p(j)$ is at most $0$.
Consider next consumers other than $i_{max}$.
By the maximality of $i_{max}$ and the price assigned to $W'_{i_{max}}$, the utility every other consumer derives from this bundle is at most $0$.
It remains to show that no consumer prefers the bundle of any other consumer other than $i_{max}$.
But this follows from the fact that the allocation and prices of these bundles remained unchanged.
As before, it remains to invoke Lemma \ref{lem:fcwe-from-partial} on the \pcbe{} to obtain a \fcwe{} with respect to \emph{all} $n$ consumers, whose revenue is at least $\opt/O(\log m)$, as required.

\subsubsection{Proof for Valuations Based on the same Matroid}

We now prove Theorem \ref{thm:matroid_rev} for markets where all valuations are based on the same matroid rank function with possibly different weights.
The proof makes use of the following definition.

\begin{definition}[Extra-consumer solution]
\label{def:extra-sol}
Suppose that all the $n$ valuations are the rank function of the same matroid $(M,\cI)$ (with possibly different weights). An \emph{extra-consumer solution} is an allocation $(S_1,\dots,S_n,S_{n+1})$ of the items to $n+1$ bundles, an item price $p_j$ for every item $j$, a bundle price $p_{S_i}$ for every bundle $S_i$ of consumer $i<n+1$, and an additional price $q$, such that the following three properties hold:
\begin{enumerate}
\item For every consumer $i<n+1$, his profit from his bundle according to the bundle prices is at least as high as his profit from any set of items $U$ according to the item prices, i.e., $\forall U\subseteq M : v_i(S_i)-p_{S_i} \ge v_i(U)-\sum_{j\in U}{p_j}$.
\item For every $i<n+1$, the bundle price $p_{S_i}$ is at least as high as the price of any independent set $U\subseteq S_i$ according to the item prices, i.e., $\forall U\in\cI \text{ s.t. }U\subseteq S_i : p_{S_i} \ge \sum_{j\in U}{p_j}$.

\item Every item $j\in S_{n+1}$ has the same price $p_j=q$, and the price of every item $j\notin S_{n+1}$ is at least as high: $p_j\ge q$.
\end{enumerate}
\end{definition}

The following lemma shows how to obtain a \fcwe{} from an extra-consumer solution. Notice that the revenue of the \fcwe{} is the revenue from the consumers in $N$ according to the bundle prices.

\begin{lemma}
\label{lem:extra-sol}
Consider an extra-consumer solution with allocation $(S_1,\dots,S_n,S_{n+1})$, item prices $\{p_j\}_{j\in M}$, bundle prices $\{p_{S_i}\}_{i=1}^{n+1}$, and an additional price $q$. If either $S_{n+1}=\emptyset$ or $q=0$ then we get a \fcwe{} as follows: the partition of items to bundles is $(B_1,\dots,B_n,B_{n+1})$ where $B_i=S_i$, the allocation of bundles to consumers is $(S_1,\dots,S_n)$, and the bundle prices are $p_{B_i}=p_{S_i}$ for $i<n+1$, where $p_{S_i}= 0$ for $i=n+1$.
\end{lemma}

\begin{proof}
First note that $B_{n+1}=S_{n+1}$ is the only unallocated bundle in the \fcwe{}, and it is either empty or has price $0$. Also note that since $q=0$ and by property (3) of the extra-consumer solution, we have a \emph{strengthened} version of property (2), which holds not only for every $i<n+1$, but also for $i=n+1$ when $p_{S_{n+1}}$ is set to zero.

Now assume for contradiction that there is a set of bundles $T$ such that consumer $i<n+1$ would prefer $T$ to his allocation $S_i$, i.e.,  $v_i(S_i)-p_{S_i} < v_i(T)-\sum_{B_\ell\in T}{p_{B_\ell}} = v_i(T)-\sum_{S_\ell\in T}{p_{S_\ell}}$ (in particular, bundle $S_{n+1}$ may belong to $T$). Consider the value $v_i(T)$. There is an independent set $U\subseteq \cup_{S_\ell \in T} S_\ell$ such that $v_i(T)=v_i(U)$. We get the following contradiction:
\begin{eqnarray*}
v_i(T) - \sum_{S_\ell\in T}{p_{S_\ell}}
&\le&
v_i(U) - \sum_{S_\ell\in T} \sum_{j\in S_\ell \cap U}{p_j}\\
& = &
v_i(U) - \sum_{j\in U}{p_j}\\
&\le &
v_i(S_i)-p_{S_i},
\end{eqnarray*}
where the first inequality is by the strengthened version of property (2), and the second inequality is by property (1) of the extra-consumer solution.
\end{proof}

Thus, our goal is to show the existence of an extra-consumer solution, in which either $S_{n+1}=\emptyset$  or the price $q$ is zero. In addition, we would like the revenue from the consumers in $N$ according to the bundle prices to be a logarithmic approximation to $\opt$. By Lemma \ref{lem:extra-sol} this is sufficient to complete the proof of the theorem. To prove existence we proceed iteratively; the next lemma shows the existence of a good starting point.

\begin{lemma}
\label{lem:starting-point}
There exists an extra-consumer solution whose revenue from the consumers in $N$ according to the bundle prices is at least $\opt/\log(m)$.
\end{lemma}

\begin{proof}
The existence of such a solution follows directly from Lemma \ref{lem:WE-high-rev-extra}, by augmenting the competitive equilibrium with additive bundle prices.
\end{proof}

We are now ready to continue the construction. We will take an extra-consumer solution with $|S_{n+1}|>0$ and $q>0$ and produce another extra-consumer solution in which either $|S'_{n+1}|<|S_{n+1}|$ or $q'=0$, making sure that the revenue did not decrease. Several such iterations are sufficient to complete the proof of existence.

Denote by $\vec{r}=(r_1,\dots,r_n)$ the ranks according to the matroid $(M,\cI)$ of the allocation $(S_1,\dots,S_n)$ of the extra-consumer solution. We consider two cases.

\paragraph{Case 1}
There is an item $j\in S_{n+1}$ and a consumer $i^*$ such that $j$ can be added to $S_{i^*}$ without increasing its rank $r_{i^*}$. If this is the case, add $j$ to $S_{i^*}$ keeping the bundle and item prices fixed. So $S'_{i^*}=S_{i^*}\cup \{j\}$, $S'_{n+1}=S_{n+1}\setminus \{j\}$ and $q'=q$. We show that the three properties of an extra-consumer solution still hold. Property (3) holds since the item prices did not change. Property (1) holds for $i\neq i^*$ since $i$'s profit did not change, and holds for $i^*$ since his profit weakly increased.
Property (2) holds for $i\neq i^*$ since $S_i$ did not change.

It is left to argue that Property (2) holds for $i^*$. Let $U\subseteq S'_{i^*}$ be an independent set in $\cI$. Assume for contradiction that $p_{S_{i^*}}<\sum_{j\in U}{p_j}$. Since we started with an extra-consumer solution, $U$ must include $j$. Let $U'\ne U$ be another independent set contained in $S'_{i^*}$ such that $|U'|=|U|$ and $j\notin U'$. Such an independent set is guaranteed to exist since the rank of $S'_{i^*}$ did not increase despite adding $j$, and is equal to the rank of $S_{i^*}$. By a well-known property of matroid independent sets, there is an item $k\in U'\setminus U$ such that $j$ can be replaced with $k$ and the set will remain independent, i.e., $U''=U\setminus\{j\}\cup\{k\}\in\cI$ \cite{Sch03}. We know that $p_{S_{i^*}}\ge\sum_{j\in U''}{p_j}$ by property (2) of the original extra-consumer solution. But $\sum_{j\in U''}{p_j} = p_k - p_j + \sum_{j\in U}{p_j}$, and by property (3) of the original extra-consumer solution, $p_k\ge p_j=q$. Thus $\sum_{j\in U''}{p_j} \ge \sum_{j\in U}{p_j}$, in contradiction to the above assumption.

\paragraph{Case 2}
For every item $j\in S_{n+1}$ and for every consumer $i$, adding $j$ to $S_i$ increases the rank $r_i$ of $S_i$. Denote $S'_i=S_i\cup\{j\}$ and let $r'_i=r_i+1$ be its rank. The increase in the rank means that $i$'s value for his bundle goes up by his weight for $j$, i.e., $v_i(S'_i)=v_i(S_i)+w_{ij}$. This is by the following argument: Clearly $i$'s value goes up by at most his weight for $j$. We know that $v_i(S_i)=v_i(U)$ for some independent set $U\subseteq S_i$, where $|U|=r_i$. We show that $U\cup\{j\}$ remains an independent set, and so $v_i(S'_i)\ge v_i(U) + v_i(\{j\})=v_i(S_i)+w_{ij}$. Assume for contradiction that $U\cup\{j\}$ is not an independent set. Then its rank $r(U\cup\{j\})$ remains $r_i$. Let $U'\subseteq S'_i$ be an independent set with size (rank) $r'_i$. By the exchange property of matroids, there is an item $k\in U'\setminus U$ that can be added to $U$ increasing its rank. By our assumption, this item cannot be $j$. But this leads to a contradiction, since $U\cup\{k\}\subseteq S_i$, and $r(U\cup\{k\})=r_i+1$.

Now let item $j^*\in S_{n+1}$ and consumer $i^*$ be such that the contribution of $j^*$ to $i^*$ is maximum over all items in $S_{n+1}$ and all consumers, i.e., $w_{i^*j^*}=\max_{i\in N,j\in S_{n+1}}\{w_{ij}\}$. Notice that $w_{i^*j^*}\le q$, otherwise property (1) of the original extra-consumer solution does not hold. Add $j^*$ to $S_{i^*}$ and raise the bundle price by $w_{i^*j^*}$, i.e., set the bundle price $p_{S'_{i^*}}=p_{S_{i^*}}+w_{i^*j^*}$. Set the item price of $j^*$ to be $p_{j^*}=w_{i^*j^*}$ and set $q'=w_{i^*j^*}$ as well. If $q'<q$, lower the prices of all items remaining in $S'_{n+1}$ to $q'$.

We now show that the three properties of an extra-consumer solution hold. Properties (2) and (3) hold by the way we modified the bundle and item prices. We first argue that property (1) holds for consumer $i^*$: His profit did not change (his value increased by $w_{i^*j^*}$ but so did his bundle price). That means that if there's an item set $U$ which he strongly prefers to his bundle $S'_{i^*}$, it must contain items other than $j^*$ whose prices decreased. The only such items are the items in $S'_{n+1}$. But since we chose $j^*$ that maximizes $w_{i^*j}$ over all $j\in S_{n+1}$, then for every $j\in S'_{n+1}$, $p_{j} = w_{i^*j^*} \ge w_{i^*j}$. This contradicts the assumption that $i^*$ strongly prefers $U$. By a similar argument, property (1) holds for every consumer $i\ne i^*$, whose profit did not change and whose value for any item with decreased price is at most the reduced price. This concludes the analysis of the second case.

\end{document}